\documentclass[%
  aip,
  jcp,
  amsmath,amssymb,
  reprint,
  longbibliography,
]{revtex4-2}

\usepackage[T1]{fontenc}
\usepackage{graphicx}
\usepackage{bm}
\usepackage{booktabs}
\usepackage[caption=false]{subfig}
\usepackage{url}
\usepackage{amsthm}
\usepackage{mathtools}
\usepackage{xcolor}
\usepackage[pdftex, colorlinks=true]{hyperref}

\colorlet{revcolor}{black}
\newcommand{\rev}[1]{\textcolor{revcolor}{#1}}
\newenvironment{revtext}{\color{revcolor}}{}

\newlength{\colfigwidth}
\makeatletter
\AtBeginDocument{\setlength{\colfigwidth}{\if@twocolumn\columnwidth\else0.5\textwidth\fi}}
\makeatother

\newtheorem{theorem}{Theorem}
\newtheorem{proposition}[theorem]{Proposition}
\newtheorem{lemma}[theorem]{Lemma}

\theoremstyle{definition}

\theoremstyle{remark}

\newcommand{\R}{\mathbb{R}}

\newcommand{\tr}{\operatorname{tr}}
\newcommand{\Eloc}{E_{\mathrm{loc}}}

\begin{document}

\title{Equivariant Continuous Normalizing Flows with Offline Sampling for Fermionic Ground State Estimation}

\author{Sam Cochran}
\email{sam@sygaldry.com}
% \email{samco@umich.edu}
\affiliation{University of Michigan, Ann Arbor, MI, USA}
\affiliation{Sygaldry Technologies Inc., Ann Arbor, MI, USA}

\author{James B. Larsen}
\affiliation{University of Michigan, Ann Arbor, MI, USA}

\author{James Stokes}
\affiliation{University of Michigan, Ann Arbor, MI, USA}

\author{Andrew Wray}
\affiliation{Sygaldry Technologies Inc., Ann Arbor, MI, USA}

\author{Michael J. Keiser}
\affiliation{Sygaldry Technologies Inc., Ann Arbor, MI, USA}

\author{Chad Rigetti}
\affiliation{Sygaldry Technologies Inc., Ann Arbor, MI, USA}

\author{Shravan Veerapaneni}
\affiliation{University of Michigan, Ann Arbor, MI, USA}
\affiliation{Sygaldry Technologies Inc., Ann Arbor, MI, USA}

\date{\today}

\begin{abstract}
We introduce a framework for fermionic \rev{variational Monte Carlo} in which a \rev{continuous normalizing flow} refines a fixed antisymmetric base wavefunction. The flow is \rev{a} permutation-equivariant neural \rev{ordinary differential equation (ODE)}, a smooth, topology-preserving map\rev{;} equivariance preserves the antisymmetry of the base, so the flow can in principle improve any antisymmetric ansatz that can be sampled efficiently. \rev{The flow deforms the nodal surface of the base but cannot change its topology, so the base must supply the correct nodal topology.} \rev{We use} Slater and Jastrow-Slater bases \rev{in this work}. Exact samples from the flow's Born distribution are obtained by pushing pre-cached base samples through the forward ODE, requiring no \rev{Markov chain Monte Carlo} at training time. \rev{Base samples are generated offline and reused across batches and runs,} enabling embarrassingly parallel training across multiple GPUs. We introduce three permutation-equivariant vector field architectures\rev{,} \emph{Pairwise Deep Sets}, \emph{FermiNet Vector Fields}, and \emph{Pairwise Deep Sets Gradient}, \rev{each balancing expressivity and computational cost differently}. We further introduce an \emph{augmented dynamics} formulation \rev{that co-evolves the derivatives required for the kinetic energy} as ODE state variables, eliminating differentiation through the ODE trajectory and \rev{significantly reducing} wall-clock time and memory.
\rev{On harmonically trapped spinless electrons, the flow reaches} ground-state energies below \rev{configuration interaction with single and double excitations, on par with a FermiNet baseline trained with the same optimizer}. \rev{Training shows} near-ideal strong scaling \rev{to} 128 NVIDIA \rev{A100 GPUs on} NERSC's Perlmutter \rev{for the largest systems tested, up to} $N = 48$ particles in three dimensions.

\end{abstract}

\maketitle

%%====================================================================
\section{Introduction}
\label{sec:intro}
%%====================================================================

Computing the ground state of fermionic many-body systems is a central problem in quantum chemistry and condensed matter physics. 
% underpinning the prediction of molecular energies, reaction barriers, and material properties from first principles. 
The fundamental challenge stems from the exponential growth of the Hilbert space dimension with the number of particles. 
For correlated systems, exact diagonalization (full configuration interaction) scales exponentially with particle number, becoming intractable beyond a handful of electrons. Coupled cluster with singles, doubles, and perturbative triples [CCSD(T)], one of the most accurate polynomial-time methods available, scales as $\mathcal{O}(N^7)$ with electron count and remains practical only for modest system sizes~\cite{szabo1996modern,coupled_cluster}. Density functional theory provides a more scalable alternative but relies on approximate exchange-correlation functionals that can fail qualitatively for strongly correlated systems, transition-metal chemistry, and van der Waals complexes \cite{dft_paper}.

\rev{Quantum Monte Carlo} methods occupy a complementary position within this hierarchy, achieving high accuracy with polynomial scaling by representing the wavefunction explicitly and evaluating high-dimensional integrals stochastically~\cite{foulkes2001}. In particular, the variational Monte Carlo (VMC) approach optimizes a parameterized wavefunction ansatz $\psi_\theta$ to minimize the energy expectation value, estimated via Monte Carlo integration over particle configurations sampled from the Born distribution $|\psi_\theta|^2 / \|\psi_\theta\|^2$. The quality of the result depends on the expressivity of the chosen ansatz. Classical ans\"atze such as Jastrow-Slater wavefunctions and backflow-transformed determinants~\cite{backflow} offer systematic improvement over mean-field theory at modest cost, but are limited in their ability to capture the many-body correlations present in strongly interacting systems.

Deep neural networks have emerged as a powerful class of VMC wavefunction ans\"atze, with the neural quantum states (NQS) framework~\cite{carleo_nqs} demonstrating that neural networks can accurately approximate high-dimensional quantum states. Architectures such as PauliNet~\cite{paulinet} and FermiNet~\cite{ferminet}, which incorporate physical symmetries and correlations directly into the network design, have achieved near-exact accuracy on small molecules within the first-quantization framework, substantially improving upon classical VMC ans\"atze. Subsequent architectures, including PsiFormer~\cite{vonglehn2023psiformer}, LapNet~\cite{lapnet}, and FIRE~\cite{fire}, have further improved accuracy and scalability.

A key challenge shared by each of these NQS approaches is the difficulty of sampling particle configurations during training. Because their Born distributions 
% $\Pi_\theta \propto |\psi_\theta|^2$ 
admit no tractable closed-form sampler, \rev{Markov chain Monte Carlo (MCMC)} sampling is a practical necessity. \rev{A single Markov chain is} inherently sequential, and \rev{the cost of equilibrating and decorrelating it increases} as the configuration space dimension grows with system size. These difficulties compound during training, since the distribution shifts with every parameter update\rev{. In practice, NQS models run many chains in parallel and advance each by only a small, fixed number of steps per update. While this approach keeps sampling cheap and parallel, it does not guarantee that a given batch will be fully equilibrated or independent of the previous one.}

% Within the context of first-quantized fermionic systems, 
Continuous normalizing flows (CNFs) offer a structurally distinct approach that avoids \rev{these difficulties}~\cite{theoretical_framework}. Where neural quantum states parameterize the entire wavefunction as a neural network, a CNF instead retains a fixed classical base wavefunction and learns only a refinement of it, building on established classical ans\"atze rather than learning the wavefunction from scratch. The base is transformed through a neural \rev{ordinary differential equation (ODE)}: a smooth, topology-preserving map that learns the correlations the base does not capture.
The antisymmetry required of fermionic wavefunctions reduces to two structural conditions: an antisymmetric base wavefunction and a permutation-equivariant ODE vector field. In principle, any antisymmetric base that can be sampled efficiently may be used; we present results using Slater and Jastrow-Slater bases.
The flow's Born distribution admits exact, independent sampling by pushing base samples forward through the ODE, rendering MCMC unnecessary at training time.

% Within the context of first-quantized fermionic systems, continuous normalizing flows (CNFs) offer a structurally distinct approach that avoids this bottleneck entirely~\cite{theoretical_framework}. A CNF defines a wavefunction ansatz by transforming a fixed base wavefunction through a learned neural ODE.
% The flow's Born distribution admits exact, independent sampling by pushing base samples forward through the ODE, rendering MCMC unnecessary at training time.
% The antisymmetry required of fermionic wavefunctions reduces to two structural conditions: an antisymmetric base wavefunction and a permutation-equivariant ODE vector field.

This paper presents three technical contributions toward a practical, large-scale realization of CNF-based VMC for continuous-space fermionic systems in three dimensions.
\emph{First}, we introduce three novel permutation-equivariant vector field architectures: \emph{Pairwise Deep Sets} (PDS), which extends Deep Sets~\cite{deep_sets, bilos} with a dedicated pairwise stream for direct access to interparticle geometry; \emph{FermiNet Vector Fields} (FVF), which adapts the multi-layer interaction structure of FermiNet~\cite{ferminet} as a CNF vector field; and \emph{Pairwise Deep Sets Gradient} (PDSG), which parameterizes the vector field as the gradient of a permutation-invariant scalar potential.
\emph{Second}, we introduce an \emph{augmented dynamics} formulation for efficient kinetic energy computation, in which all derivative quantities required for local energy evaluation are co-evolved as augmented ODE state variables, eliminating the need to differentiate through the ODE solve. This enables use of the forward Laplacian method~\cite{lapnet} within each ODE step and yields substantial savings in both wall-clock time and memory over standard forward-over-reverse automatic differentiation.
\emph{Third}, we utilize an offline sampling structure that enables data-parallel \rev{minimum-step stochastic reconfiguration (MinSR)}~\cite{minsr} training and makes richer base wavefunctions practical, since the base need only be sampled offline.
% we demonstrate this with a Jastrow-Slater base whose Jastrow factor encodes pairwise correlations prior to flow training.
% \emph{Third}, we utilize an offline sampling structure which enables data-parallel MinSR~\cite{minsr} training and supports expressive base wavefunctions, demonstrated here with a Jastrow-Slater base whose Jastrow factor encodes pairwise correlations prior to flow training.

Numerical experiments on harmonically trapped spinless electrons in three dimensions demonstrate ground-state energies surpassing \rev{configuration interaction with single and double excitations (CISD)} reference values (see Appendix~\ref{app:cisd}) for $N = 2$ to $N = 35$ electrons (configuration space up to 105 dimensions)\rev{, on par with a FermiNet baseline}. Near-ideal strong GPU scaling \rev{to} 128 NVIDIA A100s is also demonstrated for \rev{the largest systems tested, up to} $N = 48$ electrons (configuration space up to 144 dimensions).

The paper is organized as follows. Section~\ref{sec:background} reviews the VMC framework, stochastic reconfiguration, fermionic wavefunction ans\"atze, the MCMC sampling problem, CNF wavefunctions, and nodal and cusp structure. Section~\ref{sec:related} surveys related work. Sections~\ref{sec:equivariant_architectures}--\ref{sec:parallelization} describe the proposed architectures, augmented dynamics, and parallelization strategy. Section~\ref{sec:experiments} presents numerical experiments. Section~\ref{sec:discussion} discusses physical implications, limitations, and directions for future work. Section~\ref{sec:conclusion} provides concluding remarks.

%%====================================================================
\section{Background}
\label{sec:background}
%%====================================================================

\subsection{Variational Quantum Monte Carlo}
\label{sec:vmc}

Consider a quantum system of $N$ particles in $D$ spatial dimensions, with Hamiltonian
\begin{equation}
  \hat{H} = -\frac{1}{2}\sum_{i=1}^N \nabla_i^2
           + \sum_{i=1}^N V(x_i)
           + \sum_{i < j} U(|x_i - x_j|),
  \label{eq:hamiltonian}
\end{equation}
where $V$ is a one-body external potential and $U$ a pairwise interaction. The Hilbert space for this system is $L^2(\mathbb{R}^{ND})$, with $L^2$ norm $\|\psi\|^2 := \int |\psi(x)|^2\,\mathrm{d}x$. Wavefunctions are interpreted probabilistically in terms of the Born rule: the probability density of observing a particle configuration $x \in \mathbb{R}^{ND}$ upon measurement is given by the Born distribution $\Pi_\psi(x) := |\psi(x)|^2 / \|\psi\|^2$.

For any wavefunction $\psi \neq 0$, the energy is given by the Rayleigh quotient
% \begin{equation}
%   E[\psi]
%       := \frac{\langle \psi | \hat{H} | \psi \rangle}{\langle \psi | \psi \rangle}
%        = \frac{\int \psi^*(x)\,(\hat{H}\psi)(x)\,\mathrm{d}x}{\langle \psi | \psi \rangle},
%   \label{eq:rayleigh}
% \end{equation}
\begin{equation}
  E[\psi]
      := \frac{\langle \psi | \hat{H} | \psi \rangle}{\langle \psi | \psi \rangle}
       = \frac{\int \psi^*(x)\,(\hat{H}\psi)(x)\,\mathrm{d}x}{\int \psi^*(x)\,\psi(x)\,\mathrm{d}x},
  \label{eq:rayleigh}
\end{equation}
and the ground-state energy is given by
\begin{equation}
    E_0 = E[\psi_0] =  \min_{\psi \neq 0} E[\psi].
\end{equation}
Computing $E_0$ exactly requires diagonalizing $\hat{H}$, which is intractable for systems of more than a handful of particles due to the exponential growth of the Hilbert space dimension with $N$.
The variational principle motivates an alternative approach: given any $\psi \in L^2(\mathbb{R}^{ND})$, it holds that $E[\psi] \geq E_0$, with equality if and only if $\psi$ is a ground state. One can thus parameterize a family $\{\psi_\theta\}$ of wavefunctions and minimize $E[\psi_\theta]$ over $\theta$. 

Direct numerical evaluation of the $ND$-dimensional integrals in~\eqref{eq:rayleigh} is prohibitively expensive for large systems. Instead, the Rayleigh quotient can be rewritten as an expectation over the Born distribution $\Pi_\theta \propto |\psi_\theta|^2$ by introducing the local energy
\begin{equation}\label{eq:local_energy}
    \Eloc(x) := \frac{(\hat{H}\psi_\theta)(x)}{\psi_\theta(x)},
\end{equation}
so that
\begin{equation}
    E[\psi_\theta] = \frac{\int |\psi_\theta(x)|^2\,\Eloc(x)\,\mathrm{d}x}{\|\psi_\theta\|^2} = \mathbb{E}_{\Pi_\theta}\bigl[\Eloc\bigr].
    \label{eq:eloc_expectation}
\end{equation}
This formulation admits the following unbiased estimator using Monte Carlo integration over a batch of $M$ samples:
\begin{equation}
    E[\psi_\theta] \approx \frac{1}{M}\sum_{i=1}^M \Eloc(x_i), \qquad x_i \sim \Pi_\theta,
    \label{eq:eloc_mc_estimator}
\end{equation}
where each $x_i \in \mathbb{R}^{ND}$ is a particle configuration drawn from the Born distribution $\Pi_\theta$.
% and the variational principle guarantees $E[\psi] \geq E_0$, with equality if and only if $\psi$ is a ground state. VMC exploits this bound: parameterize $\{\psi_\theta\}$ and minimize $E[\psi_\theta]$ over $\theta$. Introducing the local energy
% \begin{equation}
%   \Eloc(x) := \frac{(\hat{H}\psi_\theta)(x)}{\psi_\theta(x)},
% \end{equation}
% the energy can be written as an expectation over the Born distribution
% \begin{equation}
%   E[\psi_\theta] = \mathbb{E}_{\Pi_\theta}[\Eloc],
%   \label{eq:eloc_expectation}
% \end{equation}
% admitting the unbiased Monte Carlo estimator
% \begin{equation}
%   E[\psi_\theta] \approx \frac{1}{M}\sum_{i=1}^M \Eloc(x_i),
%   \qquad x_i \sim \Pi_\theta.
%   \label{eq:eloc_mc_estimator}
% \end{equation}

The \emph{zero-variance principle} states that $\mathrm{Var}_{\Pi_\theta}[\Eloc] = 0$
if $\psi_\theta$ is an exact eigenstate of $\hat{H}$.
As a result, the noise in both the energy estimates and the parameter gradients decreases as the wavefunction approaches the exact ground state $\psi_0$. The local energy variance also serves as a key metric for evaluating convergence and wavefunction quality. 

\rev{Throughout, $\mathbb{E}_{\Pi_\theta}[\,\cdot\,]$ and $\mathrm{Var}_{\Pi_\theta}[\,\cdot\,]$ denote the mean and variance over configurations $x \sim \Pi_\theta$; where the distribution is clear from context we drop the subscript and write $\langle\,\cdot\,\rangle$ and $\mathrm{Var}[\,\cdot\,]$.}
For time-reversal symmetric systems without spin-orbit coupling, the Hamiltonian is real-valued in the position representation, and its eigenstates can be chosen to be real without loss of generality~\cite{foulkes2001}. We thus restrict to real ans\"atze $\psi_\theta\colon \mathbb{R}^{ND} \to \mathbb{R}$ for the remainder of this work. In practice, wavefunctions are represented in log-space as pairs $(\mathrm{sgn}(\psi_\theta(x)),\,\log|\psi_\theta(x)|)$ to avoid numerical over- and underflow.

\subsection{Stochastic Reconfiguration}
\label{sec:sr}

Stochastic reconfiguration (SR)~\cite{sr} is a natural gradient descent method that accounts for the geometry of the wavefunction manifold. Defining the logarithmic derivatives $O_k(x) := \partial_{\theta_k} \log|\psi_\theta(x)|$, the \rev{quantum geometric tensor} is
\begin{equation}
  S_{kl} := \mathbb{E}_{\Pi_\theta}[O_k O_l]
           - \mathbb{E}_{\Pi_\theta}[O_k]\,\mathbb{E}_{\Pi_\theta}[O_l].
\end{equation}
The SR update $S\,\delta\theta = -\nabla_\theta E$ performs gradient descent with respect to the Fisher-Rao metric on $\Pi_\theta$, \rev{the Riemannian metric induced on a family of probability distributions by its Fisher information,} rather than the Euclidean metric on parameters. The Euclidean metric on parameter space does not reflect the geometry of the wavefunction manifold: two parameter vectors nearby in Euclidean distance may define Born distributions that are far apart, while two distant parameter vectors may define nearly identical Born distributions. The Fisher-Rao metric corrects for this by measuring proximity in terms of the distributions themselves rather than the parameters. It is also invariant under \rev{reparameterization~\cite{martens2020}}, so any two parameterizations of the same wavefunction manifold yield identical SR update directions.

Given a wavefunction ansatz with $N_p$ parameters and a batch of $M$ samples $\{x_i\}_{i=1}^M \sim \Pi_\theta$, $S$ is estimated by collecting normalized mean-centered logarithmic derivatives into an $M \times N_p$ matrix $O$ constructed as
\begin{equation}
  O_{ik} = \frac{1}{\sqrt{M}}\bigl(O_k(x_i) - \langle O_k \rangle\bigr),
\end{equation}
giving the sample estimate $S \approx O^\top O$.

While SR has been shown to accelerate convergence for quantum wavefunction optimization, forming and inverting $S \in \mathbb{R}^{N_p \times N_p}$ exactly requires $\mathcal{O}(N_p^3)$ operations, which becomes prohibitive for large neural networks. 
A scalable variant called MinSR~\cite{minsr} avoids this by computing the SR update in the $M$-dimensional sample space rather than the $N_p$-dimensional parameter space. The dual matrix $T = OO^\top \in \mathbb{R}^{M \times M}$ shares the same non-zero eigenvalues as $S$,
% by the singular value decomposition, 
and assuming $M \leq N_p$, $S$ has rank at most $M$. The pseudo-inverse $S^+$ therefore contains no more information than $T^+$, and the following MinSR update is an exact reformulation of SR:
\begin{equation}\label{eq:nat_grad_update}
    \delta\theta = O^\top T^+ \varepsilon, \qquad \varepsilon_i = \frac{\Eloc(x_i) - \langle \Eloc \rangle}{\sqrt{M}},
\end{equation}
where $T^+$ denotes the pseudo-inverse with a soft cutoff on small eigenvalues. The dominant cost reduces from $\mathcal{O}(N_p^3)$ to $\mathcal{O}(M^3)$, making MinSR the more efficient formulation whenever $N_p \gg M$.

\subsection{Wavefunction ans\"atze for Fermionic Systems}
\label{sec:fermionic}

Suppose now that the particles in our system are fermions.
Writing $x = (x_1, \ldots, x_N)$ for the individual particle positions, the Pauli exclusion principle requires the wavefunction to be antisymmetric under exchange of any two particles,
\begin{equation}
    \psi(P_\sigma x) = \mathrm{sgn}(\sigma)\,\psi(x) \qquad \text{for all } \sigma \in S_N,
    \label{eq:antisymmetry}
\end{equation}
where $S_N$ is the symmetric group of permutations on $N$ elements and $P_\sigma$ permutes the particle labels according to $\sigma$. The canonical antisymmetric building block is the Slater determinant,
\begin{equation}
    \psi_{\mathrm{Slater}}(x) = \det\bigl[\phi_j(x_i)\bigr]_{i,j=1}^N
    \label{eq:slater}
\end{equation}
where each $\phi_j: \mathbb{R}^D \rightarrow \mathbb{R}$ is a single-particle orbital. Slater determinants capture the correct nodal structure of the non-interacting ground state, but cannot represent electron correlation, the many-body effects arising from particle interactions. Classical approaches for handling electron correlation include configuration interaction expansions~\cite{szabo1996modern}, backflow-transformed determinants~\cite{backflow}, and coupled-cluster methods~\cite{coupled_cluster}, each making different accuracy-cost tradeoffs.

More recently, deep neural networks have enabled a new class of highly expressive fermionic wavefunction ans\"atze. As universal approximators, neural networks are well-suited to learning arbitrary functions over the high-dimensional configuration space $\mathbb{R}^{ND}$, motivating the \rev{NQS} framework~\cite{carleo_nqs}.
% FermiNet~\cite{ferminet} and PauliNet~\cite{paulinet} build permutation-equivariant feature representations of particle configurations and read out a Slater determinant, enforcing antisymmetry by construction while allowing the network to learn arbitrary correlations. These NQS methods achieve near-exact accuracy on small molecules but require MCMC sampling from the learned Born distribution at every training iteration.
A prominent example is FermiNet~\cite{ferminet}, which builds one- and two-electron streams connected by cross-stream communication, with a Slater determinant readout that enforces antisymmetry by construction (see Appendix~\ref{app:ferminet}). FermiNet has been shown to achieve near-exact accuracy on small molecules.
Considerable innovation has followed, with many new architectures improving upon FermiNet~\cite{vonglehn2023psiformer, lapnet, fire}. However, a key challenge shared by each of these NQS methods is that the Born distribution $\Pi_\theta$ admits no tractable closed-form sampler, making an MCMC-based sampling procedure necessary.

\subsection{Markov Chain Monte Carlo Sampling}
\label{sec:mcmc}

The VMC estimator~\eqref{eq:eloc_mc_estimator} requires samples from $\Pi_\theta \propto |\psi_\theta|^2$. Direct sampling is intractable because it requires knowledge of the normalization constant $\|\psi_\theta\|^2 = \int |\psi_\theta(x)|^2\,\mathrm{d}x$, an $ND$-dimensional integral that cannot be evaluated in closed form for a general neural-network wavefunction. Markov chain Monte Carlo (MCMC) methods~\cite{metropolis,hastings} circumvent this by constructing an ergodic Markov chain whose stationary distribution is $\Pi_\theta$, requiring only pointwise evaluations of the unnormalized density $|\psi_\theta(x)|^2$ to produce samples.

In the context of VMC, however, MCMC sampling presents significant challenges for scalability. \rev{A single chain is inherently sequential: each configuration depends on the previous one, so consecutive samples are correlated and} a naive chain does not produce the independent samples assumed by the Monte Carlo estimator~\eqref{eq:eloc_mc_estimator}. A \emph{burn-in} period is required at the start of each chain to allow it to reach stationarity, and \emph{thinning} (retaining only every $j$-th sample) reduces autocorrelation between retained samples. As the system size $N$ grows and the configuration space $\mathbb{R}^{ND}$ expands, \rev{the cost of burn-in and thinning increases. Each Markov chain step requires evaluating the wavefunction over all $ND$ coordinates, and for the all-particle random walk proposals commonly used in NQS training~\cite{ferminet_github, lapnet}, the step size must shrink as the dimension grows~\cite{accept-target}, so that more steps are needed to produce unbiased and independent samples.} 

\rev{Proposals that move one particle at a time avoid the need to shrink the step size, but each pass over the configuration then takes $N$ moves, making them advantageous only where a single-particle move is cheap (such as when a low-rank update of the wavefunction is available~\cite{foulkes2001, fire}). In either case, the cost of producing unbiased and independent samples grows with $N$.}

% and these requirements become more exacting. 
Running multiple independent chains in parallel and pooling their output can mitigate \rev{this cost}. However, \rev{for samples to be unbiased,} each chain must independently complete its own burn-in, \rev{so that for a fixed total sample budget, increasing the number of chains shortens each one. This raises the fraction of compute spent on discarded samples, causing parallel efficiency to degrade}~\cite{zhao2021overcoming}.

These difficulties compound \rev{as the wavefunction is optimized, since $\Pi_\theta$ changes with each parameter update. In practice, rather than completing burn-in at every iteration, NQS models~\cite{ferminet_github, vonglehn2023psiformer, lapnet} adopt a special case of the scheme above in which each sample in the batch is assigned its own Markov chain. The chains persist across parameter updates; after an update, each chain is advanced by a small, fixed number of steps before its state is taken as a sample for the next batch. These steps accomplish some burn-in toward the updated distribution and some decorrelation from the previous batch, but guarantee neither.}

\rev{Because the chains are independent, the scheme is embarrassingly parallel and samples within a batch are uncorrelated.}
\rev{The cost is paid across iterations instead. A chain carried over from the previous iteration is distributed according to the previous Born distribution, and a fixed number of steps against the updated one does not in general bring it to stationarity. Each batch thus lags behind $\Pi_\theta$, and the energy estimates computed from it may be biased. Consecutive batches are also correlated, reducing the effective number of independent samples across iterations. Both effects shrink as the parameters converge, but during training they are controlled only indirectly, through the number of steps between updates, the Markov chain step size, and the initial burn-in length. Appropriate values for these parameters} are not known a priori and must be chosen carefully: too conservative wastes compute, while too aggressive yields correlated or biased \rev{energy} estimates. \rev{As the step size shrinks with dimension, a fixed number of steps moves each chain a shorter distance per iteration, so these choices matter more as the system grows.}

\subsection{Continuous Normalizing Flows as Wavefunction ans\"atze}
\label{sec:cnf}

\begin{figure*}
  \centering
  \includegraphics[width=\linewidth]{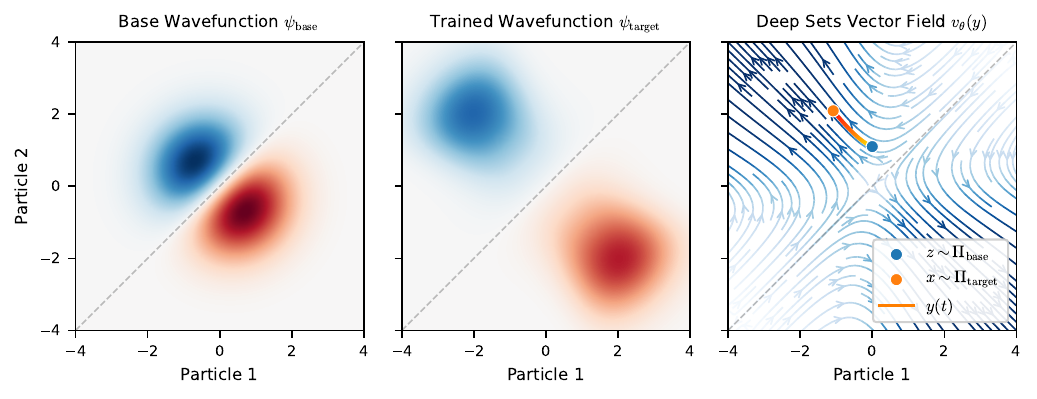}
  \caption{Visualization of a trained CNF wavefunction for two fermions in a one-dimensional harmonic trap with repulsive Coulomb interaction. \textit{Left:} Slater determinant base wavefunction $\psi_\mathrm{base}$, built from the first two harmonic oscillator orbitals; this is the exact ground state of the non-interacting system. Blue indicates positive values, red negative. Antisymmetry under particle exchange can be seen in the opposite signs of the two wavefunction lobes reflected across the Pauli exclusion diagonal $x_1 = x_2$. %: swapping $x_1$ and $x_2$ reflects across the nodal line $x_1 = x_2$ and negates $\psi$. 
  \textit{Center:} Trained wavefunction $\psi_\mathrm{target}$ after flow optimization; the increased lobe separation corresponds to the flow learning interparticle repulsion. Note that antisymmetry is preserved. \textit{Right:} Learned Deep Sets vector field $v_\theta$, overlaid with a single forward ODE trajectory $y(t)$ from a base sample $z \sim \Pi_\mathrm{base}$ (blue) to the corresponding target sample $x \sim \Pi_\mathrm{target}$ (orange). Permutation equivariance is visible as a reflective symmetry of the vector field across the diagonal $x_1 = x_2$.}
  \label{fig:flow_diagram}
\end{figure*}

Rather than defining $\psi_\theta$ directly as a neural network and relying on MCMC to draw samples, a normalizing flow defines the wavefunction through a learned transformation of a base wavefunction from which samples can be made readily available. Let $\psi_{\mathrm{base}}$ be a chosen base wavefunction and $f\colon \R^{ND} \to \R^{ND}$ a diffeomorphism. Writing $z = f^{-1}(x)$ for the preimage of a configuration $x$ under $f$, the change of variables formula gives
\begin{equation}
  \psi_{\mathrm{target}}(x)
      = \psi_{\mathrm{base}}\left(f^{-1}(x)\right)
        \left|\det\frac{\partial f^{-1}}{\partial x}\right|^{1/2},
  \label{eq:flow_wf}
\end{equation}
or equivalently in log-space,
\begin{equation}
  \log|\psi_{\mathrm{target}}(x)|
      = \log|\psi_{\mathrm{base}}(z)|
        + \tfrac{1}{2}\log\left|\det\frac{\partial z}{\partial x}\right|.
  \label{eq:log_flow_wf}
\end{equation}

The choice of parameterization for $f$ determines the tradeoff between expressivity, computational cost, and ease of inversion. Discrete normalizing flows build $f$ as a composition of finitely many invertible layers, each with an analytically tractable Jacobian determinant. A more flexible alternative is to define $f$ implicitly as the time-one map of a neural ODE: given a learned vector field $v(y)\colon \mathbb{R}^{ND} \to \mathbb{R}^{ND}$, integrating
\begin{equation}
  \frac{\mathrm{d}y}{\mathrm{d}t} = v(y), \qquad y(0) = z
  \label{eq:forward_ode}
\end{equation}
from $t = 0$ to $t = 1$ yields $x = y(1) = f(z)$. This continuous normalizing flow (CNF) formulation~\cite{neural_odes} places no structural constraints on $v$ beyond local Lipschitz continuity, which suffices for ODE well-posedness by the existence and uniqueness theorem. 
Throughout this work, $v$ is implemented using multilayer perceptrons (MLPs) with $\tanh$ activations (see Section~\ref{sec:equivariant_architectures}).
%For kinetic energy computation (Section~\ref{sec:efficient_laplacian}), Hessians of $v$ are required; the vector field is therefore implemented with $C^\infty$ activations such as $\tanh$, giving $v \in C^\infty$ and a $C^\infty$ flow map $f$ on its domain. Piecewise-linear activations such as ReLU, while locally Lipschitz, are excluded by this differentiability requirement.

The inverse $f^{-1}$ is obtained by integrating the ODE backwards in time from $x$, with no separately parameterized inverse. The log-determinant in~\eqref{eq:log_flow_wf} is tracked alongside the backwards integration via the instantaneous change of variables formula~\cite{neural_odes}, \rev{giving the augmented system}
{\color{revcolor}
\begin{alignat}{2}
  \dot{y}       &= v(y),                               &\quad y(1)       &= x,  \label{eq:inverse_ode} \\
  \dot{\lambda} &= \tr\frac{\partial v}{\partial y}(y), &\quad \lambda(1) &= 0,  \label{eq:inst_logdet}
\end{alignat}}
\rev{integrated from $t = 1$ to $t = 0$. The trace is the divergence of $v$, which measures the rate at which the flow locally expands or contracts volume, so $\lambda$ accumulates the total log-volume change along the trajectory. A single solve therefore returns $y(0) = z = f^{-1}(x)$ and $\lambda(0) = \log|\det(\partial z / \partial x)|$ together.}

\begin{figure*}
    % \centering
    \includegraphics[width=.8\linewidth]{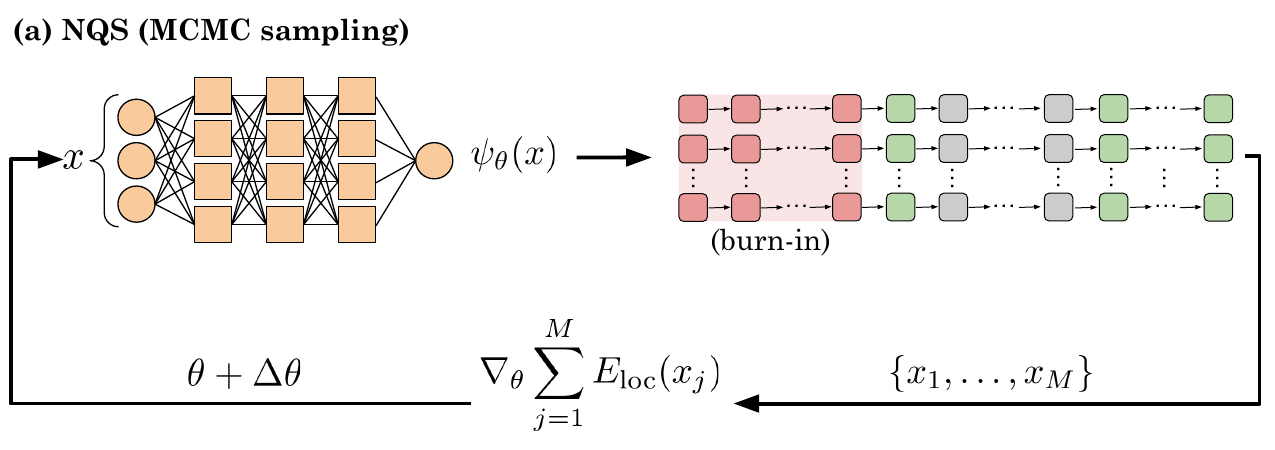}
    \includegraphics[width=.8\linewidth]{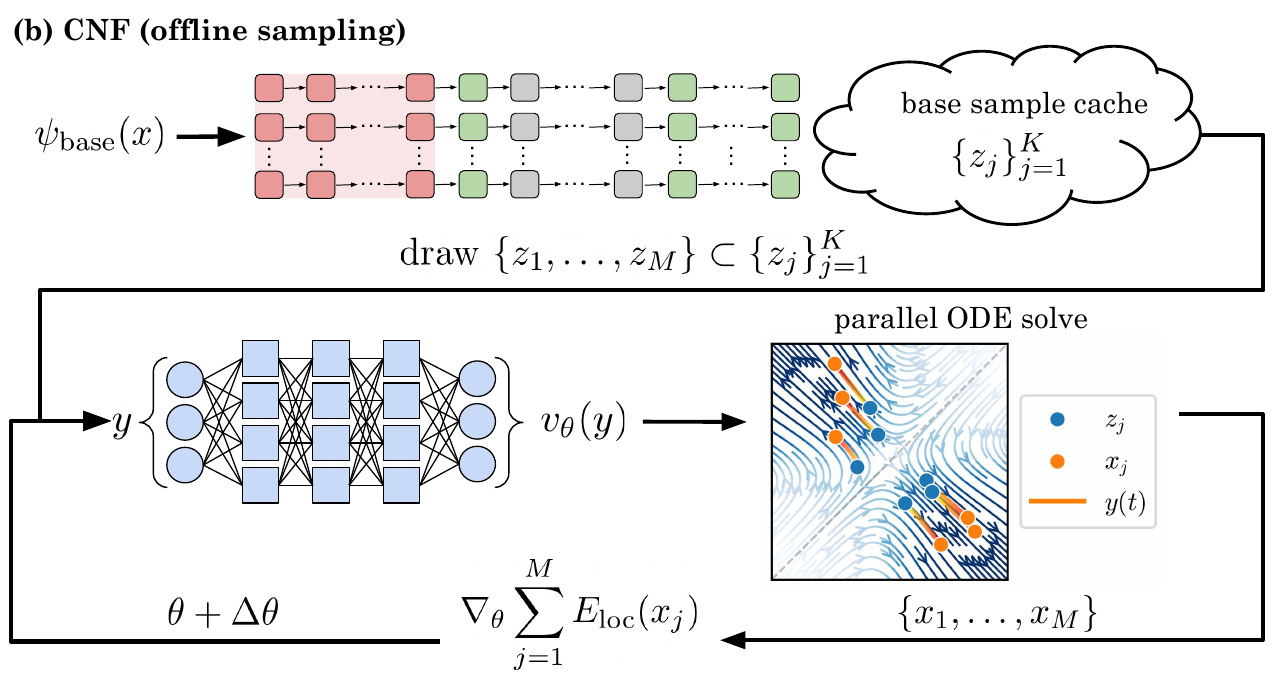}
    \caption{Comparison of the sampling procedures required for training NQS and CNF wavefunctions. NQS models (a) define the wavefunction $\psi_\theta$ directly using the output of a neural network. Sampling from the NQS Born distribution requires MCMC sampling at each training iteration\rev{; the sampling stage of (a) depicts a general MCMC procedure with several parallel chains (rows)}. Generating independent and unbiased samples $x_j$ with this method necessitates burn-in steps (red boxes) and thinning steps (gray boxes), both of which expend compute without directly producing usable samples. \rev{The procedure used in practice for NQS training is a special case, with one chain per sample in the batch. Each iteration performs only a short, fixed-length burn-in, which also serves as thinning to reduce autocorrelation between batches, and stops after a single sample per chain (the first column of green boxes); after the parameter update, the chains resume from that state (Section~\ref{sec:mcmc}).} In contrast, the CNF approach (b) uses a neural network to define a vector field $v_\theta$ whose time-one map transforms a base wavefunction to a target wavefunction. Using $M$ base Born distribution samples $z_1, \dots, z_M$ as initial conditions, samples $x_1, \dots, x_M$ from the target Born distribution can be generated by integrating the ODE $dy / dt = v_\theta(y)$ from $t = 0$ to $1$. This procedure is embarrassingly parallel given a cache $\{z_j\}_{j=1}^K$ of pregenerated base distribution samples from which the $M$ samples are uniformly drawn.}
    \label{fig:sample_diagram}
\end{figure*}

\paragraph{Exact sampling.}
If $z \sim \Pi_{\mathrm{base}}$, then $x = f(z)$ is an exact sample from $\Pi_{\mathrm{target}}$, by the change of variables formula (see Figure~\ref{fig:flow_diagram}). In practice, $f$ is computed by a numerical ODE solver, so this exactness is up to numerical solver precision. Independent base samples yield independent target samples, since $f$ is applied to each independently. Sampling thus requires only forward ODE solves as shown in Figure~\ref{fig:sample_diagram}. See Section~\ref{sec:parallelization} for a discussion of the parallel training implications.

\paragraph{Fermionic structure.}
For fermionic systems, the flow ansatz must satisfy the antisymmetry constraint~\eqref{eq:antisymmetry}. It has been shown~\cite{theoretical_framework} that this reduces to two structural requirements: (i) $\psi_{\mathrm{base}}$ is antisymmetric, and (ii) $v$ is permutation-equivariant, i.e.\ $v(P_\sigma y) = P_\sigma v(y)$ for all $\sigma \in S_N$. Equivariance of $v$ implies equivariance of the flow map $f$, and the following result establishes that these conditions are sufficient.

\begin{proposition}[Ref.~\onlinecite{theoretical_framework}, Thm.~3]
  \label{prop:antisymmetry}
  If $\psi_{\mathrm{base}}$ is antisymmetric and $v$ is permutation-equivariant, then $\psi_{\mathrm{target}}$ defined by~\eqref{eq:flow_wf} is antisymmetric.
\end{proposition}

The Slater determinant~\eqref{eq:slater} built from the lowest $N$ eigenstates of the non-interacting single-particle Hamiltonian $-\frac{1}{2}\nabla^2 + V(x)$ is a natural choice for the base wavefunction, providing the exact ground state in the absence of particle interactions. A richer alternative which builds in the electron-electron cusp and short-range repulsion is the Jastrow-Slater wavefunction $\psi_{\mathrm{JS}}(x) = J(x)\,\psi_{\mathrm{Slater}}(x)$, discussed in Section~\ref{sec:cusps}.
Both bases use the same non-interacting Slater determinant; in practice one could substitute the Hartree-Fock determinant for a lower starting energy at little added cost. We adopt the non-interacting determinant in this work to benchmark the flow, so that in the Slater case the entire interaction energy is learned by the flow.
% This factor is pre-optimized and then held fixed during flow training, as described in Section~\ref{sec:cusps}.

Before training the flow, the vector field $v$ is initialized with near-zero weights, so that the flow transformation is approximately the identity at the start of training and $\psi_{\mathrm{target}} \approx \psi_{\mathrm{base}}$. This ensures stable optimization, as training begins from whatever accuracy the base wavefunction already achieves, and the flow learns to capture residual correlations incrementally. The choice of permutation-equivariant vector field is discussed in Section~\ref{sec:equivariant_architectures}.

\subsection{Physical Wavefunction Properties: Nodal Structure and Cusps}
\label{sec:cusps}

Two physical features of the exact wavefunction have important practical consequences for variational ansatz design.

\paragraph{Nodal structure.}
The antisymmetry constraint~\eqref{eq:antisymmetry} implies that $\psi$ vanishes on a codimension-one hypersurface in $\mathbb{R}^{ND}$ called the nodal surface. The nodal surface partitions configuration space into nodal domains; the exact ground state has the fewest possible nodal domains consistent with the boundary conditions~\cite{foulkes2001}. Errors in the nodal surface directly affect the variational energy. The fixed-node diffusion Monte Carlo (DMC) method, for instance, uses the nodal surface of a trial wavefunction to constrain the \rev{imaginary-time projection}, and the fixed-node error is thus limited by the accuracy of the trial nodes. In the CNF approach, the flow transformation $f$ maps the base nodal surface to a new nodal surface. 
% in particular, the learned wavefunction can have a different nodal surface than the base. 
Since $f$ is a diffeomorphism, $\det(\partial f^{-1}/\partial x)$ is nonzero everywhere, and the Jacobian determinant factor $|\det(\partial f^{-1}/\partial x)|^{1/2}$ in~\eqref{eq:flow_wf} is strictly positive: $\psi_\mathrm{target}$ vanishes exactly where $\psi_\mathrm{base} \circ f^{-1}$ vanishes. The nodal surface of $\psi_\mathrm{target}$ is thus the image of the base nodal surface under $f$. Since $f$ is a homeomorphism, it maps connected components bijectively, preserving the number of nodal domains. 
% The number of nodal domains is therefore fixed by the base, and the flow deforms the nodal surface without altering the topology.
The number of nodal domains is therefore an invariant of the base wavefunction: the flow can reshape the nodal surface but can neither create new nodal domains nor annihilate existing ones. The base wavefunction must therefore be chosen to match the nodal topology of the target ground state. See Section~\ref{sec:discussion} for further discussion.

\paragraph{Cusp conditions.}
The Coulomb singularity in the electron-electron interaction imposes regularity conditions, known as the Kato cusp conditions~\cite{foulkes2001}, on the wavefunction at points of electron coalescence. Specifically, the spherically averaged logarithmic derivative of $\psi$ with respect to the interparticle distance $|r_{ij}|$, where $r_{ij} = x_j - x_i$, must take a prescribed value as $|r_{ij}| \to 0$. Wavefunction ans\"atze that fail to capture this short-range behavior exhibit large local energy variance near coalescence, which slows convergence of the VMC estimator. Jastrow factors are the classical remedy: a multiplicative factor $J(x) > 0$ that is symmetric under particle exchange can be designed to satisfy the cusp conditions explicitly, reducing local energy variance and accelerating convergence. In this work, the Jastrow-Slater base $\psi_\mathrm{JS}(x) = J(x)\,\psi_\mathrm{Slater}(x)$ uses a Pad\'e-Jastrow factor
\begin{equation}\label{eq:jastrow}
  J(x) = \exp \bigg(\sum_{i<j} \frac{a\,|r_{ij}|}{1 + b\,|r_{ij}|}\bigg),
\end{equation}
with the cusp parameter $a$ fixed to satisfy the electron-electron cusp condition and the repulsion strength $b$ pre-optimized via a single-parameter VMC \rev{optimization} before flow training begins. 

The diffeomorphism argument for nodal structure applies equally to cusps. Since $f$ is a diffeomorphism, the Jacobian determinant magnitude $|\det \partial f^{-1}/\partial x|$ is smooth and nonzero, and the non-smooth points of $\psi_\mathrm{target}$ are exactly the images under $f$ of the non-smooth points of $\psi_\mathrm{base}$; the flow can relocate cusps but can neither create nor destroy them. A flow wavefunction built using a Slater determinant base with smooth single-particle orbitals can therefore only approximate the correct short-range behavior near coalescence. 

The Jastrow-Slater base places a true cusp at each point of the coalescence set $\{r_{ij} = 0\}$. The flow preserves the cusp structure, and can additionally move and deform cusps; the target wavefunction can thus in principle be trained to achieve the correct cusp behavior. 
% The flow may, however, displace cusps away from the physical coalescence points or alter the logarithmic derivative at coalescence away from the analytically prescribed value; either violation introduces error and increases local energy variance. Since any such displacement or alteration raises the energy, the variational objective penalizes these errors but does not prevent them exactly. 

%%====================================================================
\section{Related Work}
\label{sec:related}

Several recent works have explored normalizing flows as wavefunction ans\"atze for VMC. Ref.~\onlinecite{lawrence} shares our motivation of exact, parallel sampling via the change-of-variables formula, demonstrating the approach in three dimensions for real-time evolution of one particle and ground-state estimation for three particles; however, their ground-state demonstrations use distinguishable particles, avoiding the issue of particle exchange symmetry entirely, with fermionic antisymmetry left as future work. The closest structural parallel to our work is Ref.~\onlinecite{finite_temp}, which applies a CNF with a Slater determinant base to finite-temperature fermionic systems, using a simple backflow-form vector field defined using scalar functions and optimized by Adam; results are presented for up to 10 electrons in 2D. Ref.~\onlinecite{m_star} extends this to 57 electrons in 2D, but relies on an architecture without a guaranteed inverse, requiring MCMC in place of exact generative sampling. Ref.~\onlinecite{waveflow} enforces antisymmetry via boundary conditions on a fundamental domain rather than through an antisymmetric base, with a proof of concept for a helium-like system of two particles in one dimension. Ref.~\onlinecite{ngairangbam2025improvedgroundstateestimation} uses a normalizing flow as a sampler for an NQS ansatz rather than as the wavefunction itself, applied to lattice quantum field theories. An alternative MCMC-free approach not based on flows is taken in Ref.~\onlinecite{autoregressive_jastrow}, wherein an autoregressive neural Jastrow factor is used alongside a Slater determinant for exact uncorrelated sampling of second-quantized wavefunctions, targeting lattice systems. Ref.~\onlinecite{zhao2021overcoming} demonstrates GPU-scalable VMC using autoregressive models with exact sampling for discrete combinatorial problems.

%%====================================================================
\section{Equivariant Vector Field Architectures}
\label{sec:equivariant_architectures}
%%====================================================================

Given an antisymmetric base wavefunction, Proposition~\ref{prop:antisymmetry} reduces the construction of an antisymmetric CNF wavefunction to the choice of a permutation-equivariant vector field. The canonical architecture for this purpose is Deep Sets (DS)~\cite{deep_sets, bilos}, defined component-wise as
\begin{equation}
  v_i(y) = g(y_i)
           + \sum_{j \neq i}^N h(y_j),
  \label{eq:deep_sets}
\end{equation}
where $v_i(y) \in \R^D$ and $g, h\colon\mathbb{R}^{D}\to\mathbb{R}^D$ are MLPs. Note that while $v\colon\R^{ND}\to\R^{ND}$, the input dimension of each network depends only on $D$, not on $N$. Unlike a network that operates on the full $ND$-dimensional configuration, neither $g$ nor $h$ compresses a growing input; increasing $N$ increases the number of network evaluations but introduces no information bottleneck. Network width and depth may still be scaled with $N$ as desired, but this is a modeling choice rather than an architectural necessity. Figure~\ref{fig:flow_diagram} shows a trained CNF with a DS vector field for two particles in one dimension.

In DS, the Jacobian trace~\eqref{eq:inst_logdet} reduces to single-particle contributions: since $h(y_j)$ for $j \neq i$ does not depend on $y_i$, only $g$ contributes to the diagonal block $\partial v_i / \partial y_i$, and
\begin{equation}
  \tr\frac{\partial v}{\partial y}
      = \sum_{i=1}^N \tr\frac{\partial}{\partial y_i} g(y_i),
  \label{eq:deep_sets_trace}
\end{equation}
where each term requires $D$ Jacobian-vector products (JVPs) through $g$ at a single-particle input~\cite{bilos}.

DS aggregates interparticle information through a single symmetric sum, which is computationally cheap but gives the vector field no direct access to interparticle displacements. For systems with distance-dependent pairwise interactions such as~\eqref{eq:hamiltonian}, this limits expressivity: the interaction energy depends explicitly on $|x_i - x_j|$, which is not recoverable from the aggregate alone. The three architectures introduced below address this directly.

\subsection{Pairwise Deep Sets}
\label{sec:pair_deep_sets}

\emph{Pairwise Deep Sets} (PDS) extends DS with a dedicated pairwise stream using a third MLP $p\colon \R^{D} \to \R^D$:
\begin{equation}
  v_i(y) = g(y_i)
           + \sum_{j \neq i}^N h(y_j) + \sum_{j \neq i}^N p(r_{ij}),
  \label{eq:deep_sets_pairs}
\end{equation}
where $r_{ij} = y_j - y_i \in \R^D$ is the displacement from particle $i$ to $j$, so each forward pass requires $\mathcal{O}(N^2)$ MLP evaluations, compared to the $\mathcal{O}(N)$ required for DS. The displacement $r_{ij}$ gives the network explicit access to each interparticle vector. The pairwise sum is permutation-equivariant, so~\eqref{eq:deep_sets_pairs} remains equivariant overall. In practice, we find that PDS constitutes a significant improvement for fermionic ground state estimation at a modest increase in cost; details are discussed in Section~\ref{sec:experiments}.

The Jacobian trace receives contributions from both $g$ and $p$:
\begin{equation}
  \tr\frac{\partial v}{\partial y}
      = \sum_{i=1}^N \tr\frac{\partial}{\partial y_i}
         g(y_i) + \sum_{i=1}^N \sum_{j \neq i}^N
           \tr\frac{\partial}{\partial y_i}
           p(r_{ij}).
  \label{eq:pairs_trace}
\end{equation}
The first sum requires $D$ JVPs per particle; the second requires $D$ JVPs per ordered pair $(i,j)$ with $i \neq j$.

\subsection{FermiNet Vector Fields}
\label{sec:ferminet_vf}

While PDS provides direct access to pairwise geometry within a single aggregation step, a more expressive approach uses multi-layer cross-stream communication to refine particle representations iteratively. \emph{FermiNet Vector Fields} (FVF) adapt the interaction layers of FermiNet~\cite{ferminet} as a permutation-equivariant vector field. The architecture uses the same one- and two-electron feature streams and interaction layers as FermiNet, but replaces the Slater determinant readout with a linear projection.
% \begin{equation}
%   v_i(x) = W_{\mathrm{out}}\,h_i^{(L)},
%   \label{eq:fvf}
% \end{equation}
% where $h_i^{(L)}\in\R^{H_1}$ are the final one-electron features and $W_{\mathrm{out}}\colon\R^{H_1}\to\R^D$.
The multi-layer cross-stream interactions give FVF access to richer interparticle structure than PDS: at each layer, the one-electron features of particle $i$ are updated using aggregated two-electron features from all pairs involving $i$, enabling information to propagate across multiple particles over successive layers. The aggregation at each layer is also done in latent space, which is strictly more expressive than the PDS approach of aggregating the outputs in $\mathbb{R}^D$. This comes at increased computational cost relative to PDS. A more detailed description of both FermiNet and FVF can be found in Appendix~\ref{app:ferminet}.

For the Jacobian trace computation, we exploit sparsity in the two-electron stream: perturbing $y_k$ affects only the $2(N-1)$ pairs $(k,j)$ and $(i,k)$ out of $N(N-1)$ total, and this sparsity is preserved through all layers. This reduces the Jacobian trace cost from the naive $\mathcal{O}(N^3 D)$ to $\mathcal{O}(N^2 D)$.

More broadly, any permutation-equivariant NQS architecture can in principle be adapted as a CNF vector field in this way, making CNF versions of other ans\"atze a natural direction for future work.

\subsection{Pairwise Deep Sets Gradient}
\label{sec:pair_deep_sets_gradient}

Each of the architectures introduced above constructs $v$ directly from neural network outputs. An alternative approach is to instead parameterize a scalar potential $\varphi\colon\R^{ND}\to\R$ and set $v = \nabla_y\varphi$, restricting the vector field to be conservative. If $\varphi$ is permutation-invariant, then $v = \nabla_y\varphi$ is permutation-equivariant. We introduce \emph{Pairwise Deep Sets Gradient} (PDSG), which takes this approach, building $\varphi$ from single-particle and pairwise permutation-invariant pooling:
\begin{align}
  \varphi(y) =& a_q\Bigl(\frac{1}{N}\sum_{i=1}^N q(y_i)\Bigr) \notag\\ &+ a_p\Bigl(\frac{1}{N(N-1)}\sum_{i=1}^N \sum_{j \neq i}^N p(r_{ij})\Bigr),
  \label{eq:pdsg_phi}
\end{align}
where $q\colon\R^{D}\to\R^{H_q}$ and $p\colon\R^{D}\to\R^{H_p}$ are embedding MLPs and $a_q\colon\R^{H_q}\to\R$ and $a_p\colon\R^{H_p}\to\R$ are scalar readout MLPs. Mean pooling is used in place of summation so that the aggregate input to each readout MLP has scale independent of $N$. The vector field is then
\begin{equation}
  v(y) = \nabla_y \varphi(y).
  \label{eq:pdsg}
\end{equation}

The gradient structure yields a key simplification for the Jacobian trace. Since $v = \nabla_y\varphi$, the Jacobian $\partial v/\partial y$ is the Hessian $\nabla^2_y\varphi$; the divergence of a gradient vector field is thus the Laplacian of its generating potential:
\begin{equation}
  \tr\frac{\partial v}{\partial y} = \nabla_y \cdot v(y) = \Delta_y\,\varphi(y).
  \label{eq:pdsg_trace}
\end{equation}
For DS, PDS, and FVF, computing $v$ and $\tr(\partial v/\partial y)$ requires two separate procedures: a forward pass for the vector field and $D$ JVPs per particle (or per pair) for the trace. For PDSG, both quantities are obtained simultaneously from a forward Laplacian~\cite{lapnet} pass on the scalar $\varphi$. The forward Laplacian method propagates values, Jacobians, and Laplacians through a network in a single forward sweep by augmenting each intermediate activation with its gradient and Laplacian with respect to the input. Applied to $\varphi\colon\R^{ND}\to\R$, one forward Laplacian pass yields $\nabla_y\varphi(y) = v(y)$ and $\Delta_y\varphi(y) = \tr(\partial v/\partial y)$ simultaneously. Table~\ref{tab:complexity} summarizes the per-ODE-step cost of each architecture.

The gradient structure of PDSG is motivated by a connection to optimal transport. OT-Flow~\cite{onken2021otflowfastaccuratecontinuous} augments the CNF training objective with a transport-cost regularization; by the Pontryagin maximum principle, the optimal time-dependent velocity field under this regularization is the gradient of a scalar potential. Our PDSG implementation is time independent, and we do not include the transport-cost regularization in our training runs, so the formal optimality result does not apply. The connection nonetheless serves as a heuristic motivation for why restricting to gradient vector fields may be beneficial. We find empirically that PDSG outperforms PDS despite the more constrained parameterization (Section~\ref{sec:experiments}).

\begin{table}[h]
  \centering
  \caption{Per-ODE-step computational cost of each vector field architecture. Forward evaluations count MLP calls; Jacobian trace reports the number of JVPs required. For PDSG, a forward Laplacian pass on $\varphi$ internally makes $\mathcal{O}(N^2)$ pairwise MLP evaluations and yields $v$ and $\tr(\partial v/\partial y)$ simultaneously, with no separate JVPs.}
  \label{tab:complexity}
  \setlength{\tabcolsep}{6pt}
  \begin{tabular}{lll}
    \toprule
    Architecture & Forward evaluations & Jacobian trace \\
    \midrule
    DS   & $\mathcal{O}(N)$            & $\mathcal{O}(ND)$ JVPs \\
    PDS  & $\mathcal{O}(N^2)$          & $\mathcal{O}(N^2 D)$ JVPs \\
    FVF  & $\mathcal{O}(N^2)$ & $\mathcal{O}(N^2 D)$ JVPs \\
    PDSG & $\mathcal{O}(N^2)$          & Fwd. Lap. on $\varphi$ \\
    \bottomrule
  \end{tabular}
\end{table}

%%====================================================================
\section{Efficient Laplacian Computation via Augmented Dynamics}
\label{sec:efficient_laplacian}
%%====================================================================

Training the flow model \eqref{eq:flow_wf} to minimize \eqref{eq:eloc_mc_estimator} requires computing local kinetic energies and parameter gradients $\nabla_\theta \Eloc(x)$. Both tasks require extracting derivative information from an ODE solve, but the two cases differ fundamentally in structure and call for different approaches. Parameter gradients are first-order derivatives of a scalar with respect to network parameters; reverse-mode automatic differentiation through the ODE solve computes all of them simultaneously in a single pass. This is handled via Diffrax~\cite{diffrax} using recursive checkpointing, which stores $\mathcal{O}(\log n)$ states during the ODE solve and recomputes intermediate states as needed during backpropagation, keeping memory sub-linear in the number of solver steps $n$.

The kinetic energy presents a more difficult challenge. For a real wavefunction $\psi$, the local kinetic energy can be written as
\begin{equation}
  T_{\mathrm{loc}}(x) = -\tfrac{1}{2}\bigl[\Delta_x\log|\psi(x)| + |\nabla_x\log|\psi(x)||^2\bigr].
  \label{eq:kinetic_local}
\end{equation} 
Evaluating~\eqref{eq:kinetic_local} for the log-space flow wavefunction \eqref{eq:log_flow_wf} requires second-order derivatives of $\log|\psi_{\mathrm{target}}|$ with respect to particle positions, which in turn require second-order derivatives of the flow transformation $f$. Unlike parameter gradients, the Laplacian of $\log|\psi_{\mathrm{target}}|$ cannot be extracted in a single pass. The standard forward-over-reverse approach obtains $\nabla_x \log|\psi_{\mathrm{target}}|$ via reverse-mode automatic differentiation through the ODE solve~\eqref{eq:forward_ode}, then applies $ND$ JVPs to extract the Laplacian. This requires materializing $ND$ reverse-mode computation graphs simultaneously. The memory requirements scale with both $N$ and network size, quickly becoming prohibitive for system sizes of interest. Computing the JVPs sequentially rather than in parallel reduces memory at the cost of proportionally longer runtime.

One possible solution is to replace~\eqref{eq:kinetic_local} with an alternative estimator derived through integration by parts~\cite{lawrence}. For real wavefunctions,
\begin{equation}
  \langle T_\mathrm{loc} \rangle = \tfrac{1}{2}\langle |\nabla_x \log|\psi||^2 \rangle,
\end{equation}
which requires only first-order derivatives and eliminates the Laplacian computation entirely. However, this estimator forfeits the zero-variance property, which is central to both optimization stability and convergence monitoring in VMC. \rev{Zero variance is a pointwise property: at an exact eigenstate $\hat{H}\psi = E\psi$ holds at every $x$, so $\Eloc(x)$ is constant. The identity above holds only in expectation, and the two estimators differ pointwise by a term whose expectation vanishes, so the integration-by-parts estimator is unbiased but has nonzero variance even at an exact eigenstate. We have observed that it produces energy estimates with substantially higher variance in practice.} \rev{We therefore} retain the standard local energy estimator~\eqref{eq:kinetic_local} and instead introduce a method which allows the needed Laplacians to be computed more efficiently.

We propose an approach in which the required derivative quantities are co-evolved through the ODE as augmented state variables alongside the particle positions. This \emph{augmented dynamics} formulation eliminates differentiation through the ODE solve when evaluating the kinetic energy, yielding significant savings in both time and memory \rev{(see Section~\ref{sec:aug_benchmark})}.

From~\eqref{eq:log_flow_wf}, the gradient and Laplacian of $\log|\psi_{\mathrm{target}}|$ decompose in terms of six quantities associated with the inverse map and its log-determinant:
\begin{equation}
\renewcommand{\arraystretch}{1.4}
\begin{array}{r@{{}:={}}l@{\qquad}r@{{}:={}}l@{\qquad}r@{{}:={}}l}
  z      & f^{-1}(x),      & J      & \tfrac{\partial z}{\partial x}, & L    & \Delta_x z,      \\
  \lambda & \log|\det J|, & \gamma & \nabla_x\lambda,                & \ell & \Delta_x\lambda.
\end{array}
\label{eq:laplacian_decomp_pieces}
\end{equation}
Writing $s := \log|\psi_{\mathrm{base}}|$, it holds that
\begin{align}
\begin{split}
  \nabla_x \log|\psi_{\mathrm{target}}|
      &= J^\top \nabla_z s + \tfrac{1}{2}\,\gamma, \\
  \Delta_x \log|\psi_{\mathrm{target}}|
      &= \tr\left(J^\top H_s J\right) + \nabla_z s \cdot L + \tfrac{1}{2}\,\ell,
  \label{eq:laplacian_decomp}
\end{split}
\end{align}
where $H_s$ is the Hessian of $s$ at $z$; the first line follows from the standard gradient chain rule and the second from the Laplacian chain rule stated in Appendix~\ref{app:laplacian_chain_rule}. The quantities~\eqref{eq:laplacian_decomp_pieces}, together with $\nabla_z s$ and $H_s$, suffice to evaluate the local kinetic energy~\eqref{eq:kinetic_local} via~\eqref{eq:laplacian_decomp}.

The following proposition (proven in Appendix~\ref{app:augmented_dynamics}) gives governing equations for the six quantities \eqref{eq:laplacian_decomp_pieces} as augmented ODE state variables.

\begin{proposition}
\label{prop:augmented_dynamics}
Let $J_v := \partial v / \partial y_1$, $H_{v_k}$ the Hessian of the $k$-th component of $v$ with respect to $y_1$, and $H_{\tr J_v}$ the Hessian of $\tr J_v$ with respect to $y_1$. The following initial value problem, integrated from $t = 1$ to $t = 0$, yields the six quantities of~\eqref{eq:laplacian_decomp_pieces} as $y_1(0) = z$, $y_2(0) = J$, $y_3(0) = L$, $y_4(0) = \lambda$, $y_5(0) = \gamma$, $y_6(0) = \ell$:
\begin{alignat}{2}
  \dot{y}_1 &= v(y_1),                                                                       &\quad y_1(1) &= x,           \label{eq:aug_dy} \\
  \dot{y}_2 &= J_v\, y_2,                                                                    &\quad y_2(1) &= I,           \label{eq:aug_dJ} \\
  \dot{y}_3 &= J_v\, y_3 + \bigl[\tr(y_2^\top H_{v_k} y_2)\bigr]_{k=1}^{ND},               &\quad y_3(1) &= 0,           \label{eq:aug_dL} \\
  \dot{y}_4 &= \tr J_v,                                                                      &\quad y_4(1) &= 0,           \label{eq:aug_dlambda} \\
  \dot{y}_5 &= y_2^\top \nabla_{y_1}(\tr J_v),                                              &\quad y_5(1) &= 0,           \label{eq:aug_dgamma} \\
  \dot{y}_6 &= \tr\left(y_2^\top H_{\tr J_v} y_2\right) + \nabla_{y_1}(\tr J_v)\cdot y_3, &\quad y_6(1) &= 0.           \label{eq:aug_dl}
\end{alignat}
Here $y_1, y_3, y_5 \in \R^{ND}$, $y_2 \in \R^{ND\times ND}$, and $y_4, y_6 \in \R$. These solutions at $t=0$ can be assembled to evaluate the local kinetic energy~\eqref{eq:kinetic_local} via~\eqref{eq:laplacian_decomp}; $y_4(0)$ additionally provides the log-determinant for wavefunction evaluation via~\eqref{eq:log_flow_wf}.
\end{proposition}

At each solver step, the right-hand sides of~\eqref{eq:aug_dy}--\eqref{eq:aug_dl} are evaluated by applying the forward Laplacian method~\cite{lapnet} to the pair $(v,\, \tr J_v)$. This propagates values, gradients, and second-order terms for both functions simultaneously in a single forward pass, yielding $v$, $J_v$, and $H_{v_k}$ from the first component and $\tr J_v$, $\nabla_{y_1}(\tr J_v)$, and $H_{\tr J_v}$ from the second, with no differentiation through the ODE solve.
% The forward Laplacian operates on individual network evaluations, not on compositions of ODE steps; applying it through the full ODE solve would require composing second-order information across many solver steps, which it is not designed to do. In the augmented dynamics formulation it is applied once per step to evaluate $v$ and its derivatives at the current state, which is exactly its intended use.
The dominant memory cost of the augmented state comes from storing the Jacobian $y_2 \in \R^{ND \times ND}$~\eqref{eq:aug_dJ} \rev{for every sample in the batch at each of the $\mathcal{O}(\log n)$ recursive checkpoints}; because the kinetic energy assembles algebraically from the augmented ODE solutions via~\eqref{eq:laplacian_decomp} and~\eqref{eq:kinetic_local}, no storage of network activations is needed for the kinetic energy computation. \rev{Jacobian storage is the binding memory cost of the approach at large $N$, though the peak memory of the approach remains below that of forward-over-reverse differentiation in either mode at every size tested (Section~\ref{sec:aug_benchmark}).}

%%====================================================================
\section{Offline Sampling and Parallelization}
\label{sec:parallelization}
%%====================================================================

\begin{figure}
  \centering
  \includegraphics[width=0.9\colfigwidth]{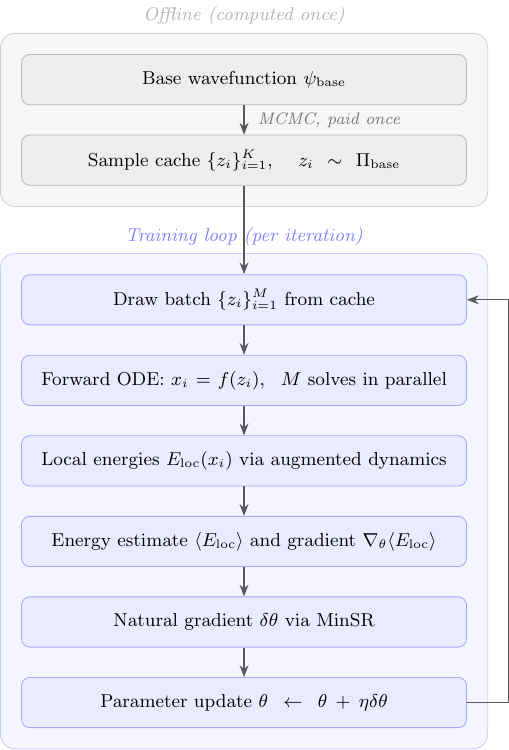}
  \caption{Training pipeline for the CNF-based VMC framework. The offline phase (gray) runs MCMC once on the fixed base distribution $\Pi_\mathrm{base}$ to populate a sample cache of $K$ configurations. Each training iteration (blue) draws a random batch from the cache, propagates each sample through the forward ODE to obtain exact target samples, computes local energies via augmented dynamics, and applies a MinSR natural gradient update to the parameters $\theta$. No MCMC is performed at training time.}
  \label{fig:pipeline}
\end{figure}

As discussed in Section~\ref{sec:mcmc}, NQS models must perform online MCMC sampling at every training iteration. \rev{The scheme used in practice gains parallelism at the cost of batches that may lag behind the current Born distribution and are correlated across iterations. The CNF approach avoids this trade-off} by relocating MCMC sampling entirely offline (see Figure~\ref{fig:sample_diagram}).

Prior to training, a cache of $K$ base samples $\{z_i\}_{i=1}^K$ is generated by running MCMC on the base distribution $\Pi_{\mathrm{base}}$. This is a task MCMC is well-suited for: the burn-in cost is paid once and amortized over the full cache, and the chains can be run for long enough to ensure thorough mixing and decorrelation. \rev{The base can also be sampled with whichever proposal suits it best. For Slater and Jastrow-Slater bases, single-particle moves with low-rank updates are available~\cite{foulkes2001}, and because MCMC is applied only to the base, this efficiency does not constrain the architecture of the flow.} Because the base wavefunction is fixed throughout training, the sample cache remains valid for the entire training run. Furthermore, the sample cache for a given base wavefunction needs to be generated only once, and can then be reused indefinitely across multiple flow models and training runs. By saving the Markov chain state and MCMC hyperparameters along with the samples, additional samples can be added to the cache as needed using already burned-in chains.

At each training iteration, a batch $\{z_i\}_{i=1}^M$ is drawn from the cache. Each base sample is propagated through the forward ODE~\eqref{eq:forward_ode} to yield $x_i = f(z_i) \sim \Pi_{\mathrm{target}}$. Since each forward ODE solve depends only on $z_i$ and the current parameters $\theta$, the $M$ solves are independent and can be distributed across GPUs with no inter-sample communication beyond the natural gradient computation and parameter update~\eqref{eq:nat_grad_update}. As the flow parameters are updated, the same cached base samples are mapped to different configurations by the updated flow transformation $f$, so the sampling cost does not grow with training. \rev{Within each pass through the cache, base samples are drawn without replacement, so a given $z_i$ reappears only after $K/M$ iterations, by which point the parameters, and hence the configuration $f(z_i)$ it is mapped to, have changed. The cache should nonetheless be large enough that samples are not replayed often during training, and that it gives a representative sampling of $\Pi_\mathrm{base}$.}
% This offline structure is the key structural advantage of flow-based VMC over MCMC-dependent NQS approaches\rev{: it matches the parallelism of the sampling scheme used in practice while supplying exact samples from the current Born distribution at every iteration}. It underlies the multi-GPU scaling demonstrated in Section~\ref{sec:experiments}.

This offline structure underlies the multi-GPU scaling demonstrated in \rev{Sections~\ref{sec:scaling} and~\ref{sec:weak_scaling}}. The full training pipeline is summarized in Figure~\ref{fig:pipeline}.

%%====================================================================
\section{Experimental Results}
\label{sec:experiments}
%%====================================================================

\subsection{Experimental Setup}

\paragraph{Physical system.}
All experiments consider $N$ spinless electrons confined in a three-dimensional harmonic trap with Hamiltonian
\begin{equation}
  \hat{H} = -\frac{1}{2}\sum_{i=1}^N \nabla_i^2
            + \frac{1}{2}\omega^2 \sum_{i=1}^N |x_i|^2
            + k \sum_{i < j} \frac{1}{|x_i - x_j|},
  \label{eq:harmonic_hamiltonian}
\end{equation}
with $\omega = 1$ and $k = 1$ in Hartree atomic units. For a sense of scale, the $N=2$ Hartree-Fock wavefunction has a root-mean-square radius of about 1.5 Bohr radii, comparable to the spatial extent of valence electrons in real atoms. The harmonic trap is a canonical model for confined quantum systems such as ultracold trapped atoms and quantum dots. The following four properties also make it well-suited as a benchmark for first-quantized VMC methods. First, its non-interacting limit is exactly solvable, giving a well-motivated antisymmetric base wavefunction that the flow refines. Second, the trap frequency $\omega$, interaction strength $k$, and particle number $N$ tune the correlation strength and system size, yielding a controllable family of test problems. Third, the configuration space is the unbounded, smooth domain $\mathbb{R}^{3N}$, with the confinement isolating electron-electron correlation without electron-nuclear singularities or boundary and periodicity conditions. Fourth, the system admits trustworthy reference energies for validation: a numerically exact solution can be obtained for the two-electron case, and for larger $N$ we can obtain CISD reference energies using a Gaussian basis that spans the harmonic-oscillator eigenstates. All reported CISD reference energies use a shell cutoff of $n_{\max} = 10$ ($N_{\mathrm{orb}} = 286$ orbitals; see Appendix~\ref{app:cisd}).

\paragraph{Base wavefunctions.}
The Slater base uses as orbitals the lowest $N$ eigenstates of the single-particle harmonic oscillator,
\begin{equation}
  \phi_{n}(x) \propto \prod_{d=1}^{3} H_{n_d}(\sqrt{\omega}\,x_d)\,e^{-\omega x_d^2/2},
  \label{eq:ho_orbitals}
\end{equation}
where $n \in \mathbb{N}_0^3$ indexes quantum numbers ordered by $|n|$ and $H_{n_d}$ is the $n_d$-th Hermite polynomial. This Slater determinant is the exact non-interacting ground state. The Jastrow-Slater base augments this with the Pad\'e-Jastrow factor~\eqref{eq:jastrow}, with $a = k/4$ and $b$ pre-optimized by a single-parameter VMC \rev{optimization}, after which the Jastrow-Slater base is frozen and samples are drawn and cached.

\paragraph{CNF architectures.}
Four permutation-equivariant vector fields are evaluated: DS, PDS, FVF, and PDSG, each defining a CNF wavefunction. Architecture, solver, and training hyperparameters, along with parameter counts, are given in Appendix~\ref{app:hyperparameters}.

\paragraph{Training.}
All models were trained using a batch size of 4096. Base samples were drawn from a pregenerated cache of size $2^{19}$ and distributed evenly across GPUs\rev{; all CNF runs converged within three passes through the cache}.
Optimization was carried out using MinSR~\cite{minsr}. Convergence was declared once both the local energy mean and variance stopped decreasing; the reported energy at each system size was taken at the lowest-variance checkpoint in the final converged window. The full convergence criterion is detailed in Appendix~\ref{app:hyperparameters}. \rev{As an NQS baseline, a FermiNet wavefunction was trained on the same system with the same batch size, optimizer, and convergence criterion, sampled with the persistent-chain scheme described in Section~\ref{sec:mcmc}. Its architecture is described in Appendix~\ref{app:ferminet}, and its sampling and pretraining hyperparameters are given in Appendix~\ref{app:hyperparameters}.}

\paragraph{Compute.}
All experiments except the scaling results (\rev{Sections~\ref{sec:scaling} and~\ref{sec:weak_scaling}}) were executed via Modal on NVIDIA A100 GPUs, using a single GPU for the augmented dynamics benchmark (Section~\ref{sec:aug_benchmark}) and up to 8 GPUs for the training experiments. The scaling results were run on the Perlmutter machine at NERSC.

\subsection{Augmented Dynamics Benchmark}
\label{sec:aug_benchmark}

\begin{figure}
  \centering
  \includegraphics[width=\colfigwidth]{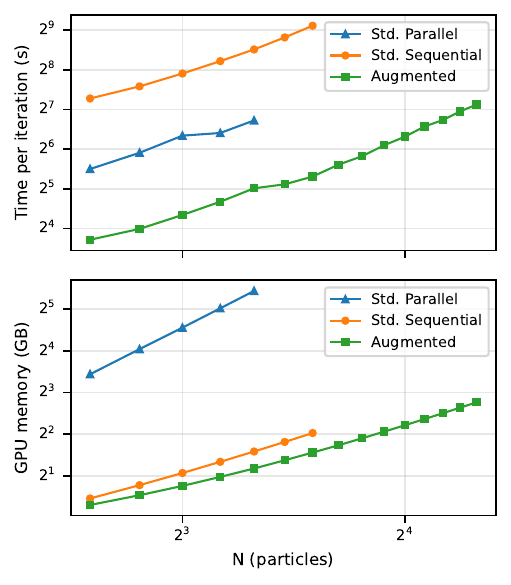}
  \caption{Per-iteration wall-clock time (top) and peak GPU memory usage (bottom) as a function of particle number $N$, comparing the augmented dynamics approach introduced in Section~\ref{sec:efficient_laplacian} with standard forward-over-reverse automatic differentiation through the ODE solve, in both parallel and sequential modes. The ODE solve was fixed to use 10 steps, and runs were carried out on a single GPU using a PDS vector field over a Slater base. \rev{Results are reported at \rev{each} $N$ from 6 to 20 for augmented dynamics, from 6 to 12 for the standard sequential variant, and from 6 to 10 for the standard parallel variant, which exhausted the available memory beyond $N = 10$; the other two ranges reflect the compute budget allotted rather than any limit of the methods.} The augmented approach is faster and more memory-efficient than both standard variants across all system sizes tested.}
  \label{fig:aug_benchmark}
\end{figure}

Figure~\ref{fig:aug_benchmark} compares per-iteration cost across three kinetic energy evaluation strategies: the augmented dynamics formulation of Section~\ref{sec:efficient_laplacian}, and standard forward-over-reverse automatic differentiation in its parallel and sequential JVP modes. The ODE solve was fixed to 10 steps so that the trend in $N$ reflects per-step cost rather than differences in adaptive step counts. Augmented dynamics is faster and more memory-efficient than standard forward-over-reverse differentiation in both its parallel and sequential variants, across all system sizes tested. The improvement is most pronounced where each variant incurs its dominant cost. The memory advantage is largest over the parallel variant, which materializes $ND$ adjoint computation graphs simultaneously, while the runtime advantage is largest over the sequential variant, which replays the trajectory differentiation $ND$ times. Augmented dynamics avoids both costs by assembling the Laplacian algebraically from the augmented ODE state, storing only $\mathcal{O}(N^2 D^2)$ plain arrays per checkpoint.

\subsection{Comparison of Vector Field Architectures}

\begin{figure*}
  \centering
  \subfloat[\rev{$N = 10$}\label{fig:convergence_n10}]{%
    \includegraphics[width=0.5\textwidth]{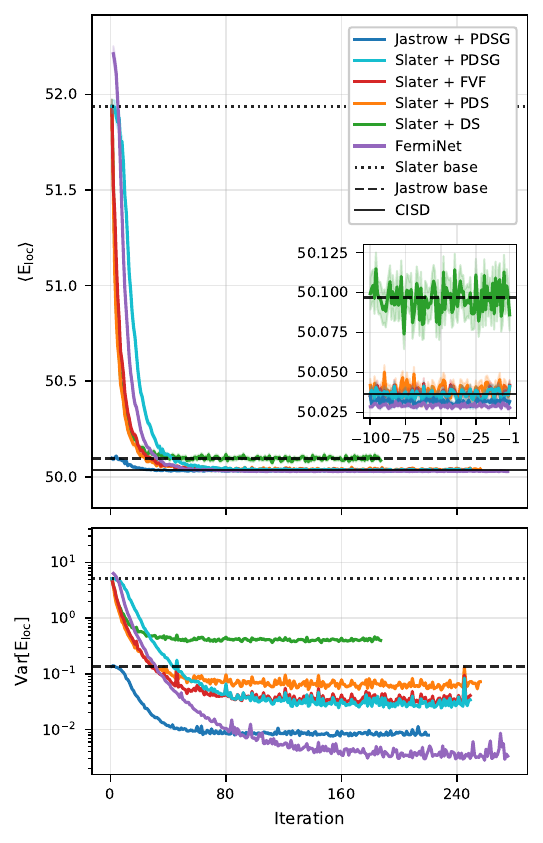}}%
  \hfill
  \subfloat[\rev{$N = 20$}\label{fig:convergence_n20}]{%
    \includegraphics[width=0.5\textwidth]{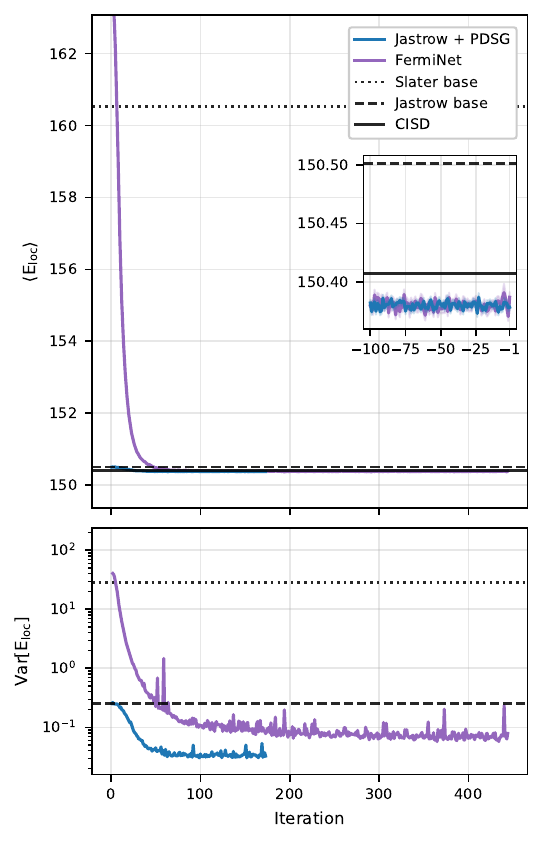}}%
  \caption{\rev{Convergence of CNF wavefunctions and the FermiNet baseline for (a) $N = 10$ and (b) $N = 20$ harmonically trapped electrons in three dimensions, trained on 2 and 4 GPUs respectively. At $N = 10$, results are shown for each vector field architecture using a Slater base wavefunction, and for PDSG with a Jastrow-Slater base; at $N = 20$, for PDSG with a Jastrow-Slater base.}
  The upper panels show the convergence of the variational energy estimate $\langle \Eloc \rangle$. Shaded regions indicate $\pm 1$ standard error. The energies of the two base wavefunctions are plotted as horizontal lines, along with a CISD reference energy. \rev{The insets show} energies over the final 100 iterations of training for each model, corresponding to the convergence window.
  The lower panels show the convergence of the local energy variance \rev{$\mathrm{Var}[\Eloc]$, with the variances of the two base wavefunctions plotted as horizontal lines}.
  \rev{Jastrow-Slater PDSG is the most accurate flow model; FermiNet reaches the lowest local energy variance at $N = 10$, while the flow reaches the lowest local energy variance at $N = 20$.}
  }\label{fig:convergence}
\end{figure*}

Figure~\ref{fig:convergence}\rev{(a)} compares the four vector field architectures on a system of $N = 10$ harmonically trapped electrons in three dimensions. Of the \rev{flow} models trained using a Slater base, PDSG achieves the lowest converged energy and local energy variance values, followed by FVF, PDS, and then DS. The separation is clearest in the local energy variance\rev{, which by the zero-variance principle (Section~\ref{sec:vmc}) measures wavefunction quality without reference to an external benchmark}. Relative to the CISD reference energy, the three most accurate architectures converge close to CISD, while DS remains well above it.

\subsection{Choice of Base Wavefunction}

Having identified PDSG as the strongest architecture on the Slater base, we also trained a PDSG CNF with a Jastrow-Slater base. Both the converged energy and the local energy variance improve substantially over the Slater-base PDSG model, and training also converges \rev{in fewer iterations}. Two factors contribute: a CNF with a Jastrow-Slater base begins training from a lower energy, and it builds in the correct electron-electron cusp behavior, leaving less residual correlation for the flow to capture. Notably, the Jastrow-Slater PDSG model converges below the CISD reference energy.
% Because both CISD and the variational CNF energy are upper bounds on the true ground-state energy, this lower energy is a strictly better estimate and provides strong evidence for the quality of the model.
% The importance of the base is further underscored by the DS model, which on the Slater base converges only to an energy comparable to the Jastrow-Slater base at initialization.

\begin{revtext}
\subsection{Comparison with FermiNet}

At $N = 10$ (Figure~\ref{fig:convergence_n10}), FermiNet reaches the lowest energy and local energy variance of all models, below the Jastrow-Slater PDSG model on both. At $N = 20$ (Figure~\ref{fig:convergence_n20}), where only these two models were trained, the ordering reverses: the two converge to comparable energies, but the Jastrow-Slater PDSG model reaches a lower variance and converges in fewer iterations. Both models converge below the CISD reference at both system sizes.
\end{revtext}

\subsection{Per-Model Timing}

\begin{figure}
  \centering
  \includegraphics[width=\colfigwidth]{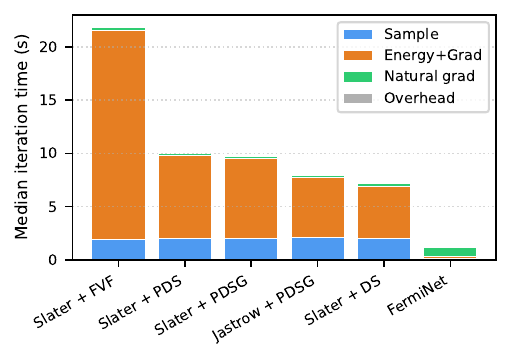}
  \caption{Median per-iteration wall-clock time for the training runs shown in Figure~\ref{fig:convergence}, with $N = 10$ harmonically trapped electrons and 2 GPUs. Cost is broken down by sampling time, energy and gradient computation time, natural gradient computation, and remaining per-iteration overhead. 
  The dominant cost \rev{of the flow models} is the energy and gradient computation, which varies greatly across \rev{architectures}. Of the models trained using a Slater base, DS is cheapest, followed by PDSG, then PDS, then FVF. The PDSG model becomes faster per-iteration when trained with a Jastrow-Slater base.}
  \label{fig:timing}
\end{figure}

For the $N = 10$ system, Figure~\ref{fig:timing} breaks down the median per-iteration wall-clock time for each model into sampling, energy and gradient computation, natural gradient computation, and remaining overhead. The dominant cost for each \rev{flow} model is the energy and gradient computation, which varies greatly across architectures. Among the Slater-base models, DS is the cheapest, followed by PDSG, then PDS, then FVF; FVF is the most expensive, owing to its multi-layer cross-stream interactions and larger effective feature dimension per layer. Training the PDSG model with a Jastrow-Slater base rather than a Slater base reduces its per-iteration cost substantially. This reduction can be traced to the ODE solve: after an initial warm-up period of training, the Slater-base models take between four and seven adaptive Tsit5 steps per solve, whereas the Jastrow-Slater PDSG model never takes more than one. Because the Jastrow factor already encodes the electron-electron cusp and short-range repulsion, the flow transformation is milder and closer to the identity, making the ODE substantially easier to integrate. The per-iteration cost of the flow models is dominated by the ODE solve, which can be reduced by increasing the GPU count \rev{or by using a lower-order solver with looser error tolerances. FermiNet is roughly an order of magnitude faster per iteration than the flow models, though it takes more iterations to converge.}

\subsection{\rev{Strong} GPU Scaling}\label{sec:scaling}

\begin{figure*}
  \centering
  \includegraphics[width=\textwidth]{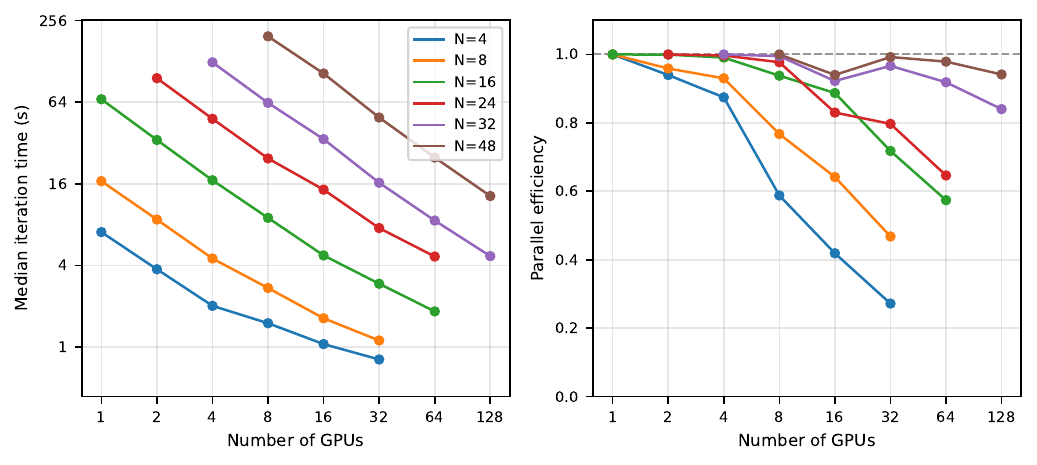}
  \caption{\rev{Strong scaling of the} per-iteration wall-clock time \rev{(left) and parallel efficiency (right)} as a function of GPU count, for the PDS CNF model using a Slater base. Each time is the median over the first 10 training iterations. Results are shown for $N \in \{4, 8, 16, 24, 32, 48\}$ electrons in three dimensions on up to 128 NVIDIA A100s, fixing 10 ODE steps per iteration \rev{and a total batch size of 4096}.
  \rev{Parallel efficiency remains near ideal to 128 GPUs for the largest systems, and falls at small $N$, where the work per GPU is small.}
  }
  \label{fig:scaling}
\end{figure*}

\begin{figure}
  \centering
  \includegraphics[width=\colfigwidth]{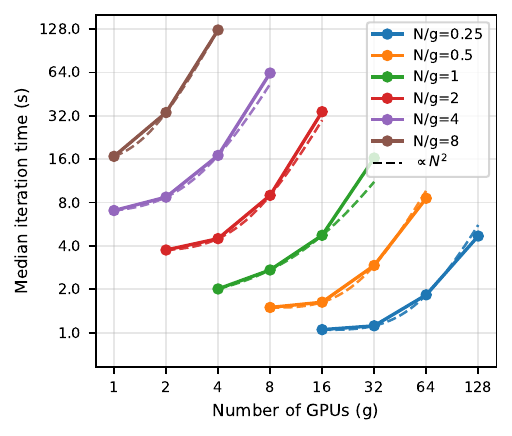}
  \caption{\rev{Weak scaling of the per-iteration wall-clock time as a function of GPU count $g$, for the PDS CNF model using a Slater base. Each time is the median over the first 10 training iterations. Results are shown for fixed ratios $N/g \in \{0.25, 0.5, 1, 2, 4, 8\}$ on up to 128 NVIDIA A100s, fixing 10 ODE steps per iteration and a total batch size of 4096. Dashed lines show fits of the form $a + bN^2$ to each series.
  At fixed $N/g$, the per-iteration time grows quadratically with $N$, reflecting the $\mathcal{O}(N^3)$ total cost of the augmented dynamics with the PDS vector field.}}
  \label{fig:weak_scaling}
\end{figure}

Figure~\ref{fig:scaling} reports strong scaling results for the PDS CNF model, measured on GPU-accelerated nodes of the Perlmutter machine at NERSC, an HPE Cray EX supercomputer with four NVIDIA A100 GPUs and one AMD EPYC 7763 CPU per node. The model was tested on \rev{up to} 128 GPUs for $N \in \{4, 8, 16, 24, 32, 48\}$, with the batch size of 4096 distributed evenly across GPUs in all cases.

\rev{At the largest system sizes,} the PDS CNF model achieves near-ideal strong scaling, with per-iteration time roughly halving at each doubling of GPU count \rev{and parallel efficiency above 0.9 at 128 GPUs for $N = 48$}.

Two departures from ideal scaling appear. First, \rev{parallel efficiency drops} at the 4-to-8 GPU boundary, where execution first spans multiple nodes and incurs inter-node communication. This \rev{drop} is seen most clearly in the $N=4$ curve. Second, the smaller-$N$ scaling curves begin to flatten as the number of GPUs increases. This flattening is consistent with Amdahl's law: at fixed batch size, the parallelizable ODE work per GPU shrinks as GPUs are added, until the fixed serial and communication costs limit further speedup. Smaller systems, whose per-iteration compute is smaller, reach this limit at lower GPU counts. Both of these effects become negligible for our larger-$N$ runs.

Model parameters are replicated across all GPUs, so communication is required only to synchronize parameters at each update step. The natural gradient (MinSR) solve is not parallelized, but it accounts for a small fraction of the per-iteration cost (see Figure~\ref{fig:timing}), and its cost is set by the number of samples rather than by $N$.
% , since the MinSR solve inverts a matrix whose dimension equals the number of samples. 
As $N$ grows, the parallelizable ODE work per iteration increases while these serial and communication costs remain essentially fixed, so the compute-to-communication ratio improves. We therefore expect near-ideal scaling to extend to higher GPU counts as $N$ increases, provided the number of samples is held fixed.

\subsection{\rev{Weak GPU Scaling}}\label{sec:weak_scaling}

\rev{Figure~\ref{fig:weak_scaling} reports weak scaling results for the same PDS CNF model, with $N$ increased in proportion to the GPU count $g$. The model was tested on up to 128 GPUs at fixed ratios $N/g \in \{0.25, 0.5, 1, 2, 4, 8\}$, with the batch size of 4096 distributed evenly across GPUs in all cases.}

\rev{At each fixed ratio $N/g$, the per-iteration time grows quadratically with $N$, and is well described by $a + bN^2$, where $a$ captures fixed per-iteration overhead. This growth is a direct consequence of the cost of the augmented dynamics (Section~\ref{sec:efficient_laplacian}): each sample carries $ND$ tangent directions through the PDS vector field, whose evaluation costs $\mathcal{O}(N^2)$, so the total work per iteration grows as $\mathcal{O}(N^3)$. Distributed across $g \propto N$ GPUs, the work per GPU grows as $\mathcal{O}(N^2)$. The growth in per-iteration time therefore reflects the cost of the method rather than a loss of parallel efficiency.}

\subsection{Converged Energies Across System Sizes}

\begin{figure*}
  \centering
  \includegraphics[width=\textwidth]{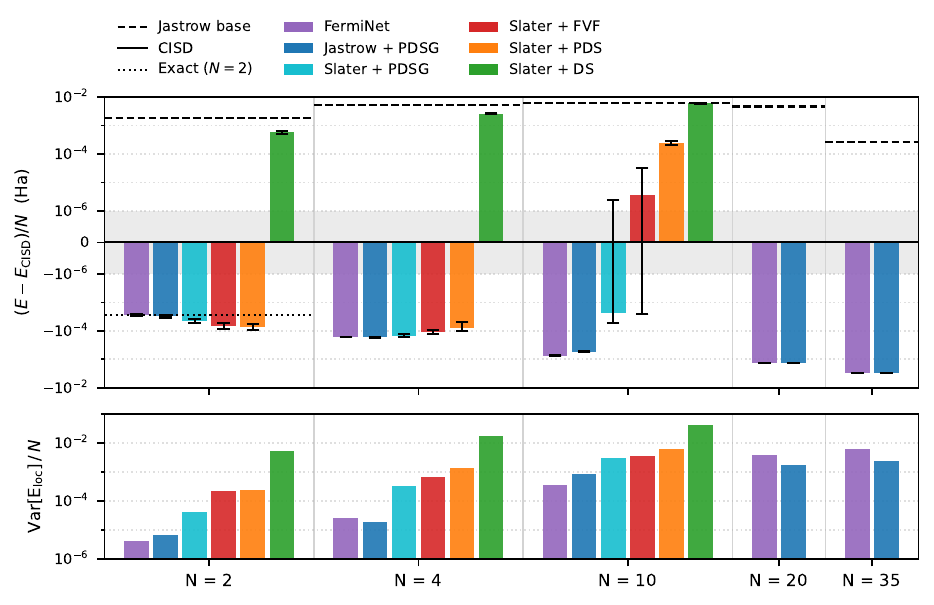}
  \caption{\rev{\textit{Upper:}} converged per-particle energy \rev{of each model} relative to the CISD reference $(\rev{E} - E_\mathrm{CISD})/N$, each estimated from an inference batch of \rev{$M = 409{,}600$} samples\rev{, drawn from the trained flow using held-out base samples for the CNF models and by MCMC for FermiNet, with error bars indicating} $\pm\,\mathrm{stderr}/N$.
  Each \rev{flow} model was trained using a Slater base on systems of $N \in \{2, 4, 10\}$ electrons; PDSG was additionally trained using a Jastrow-Slater base for $N \in \{2, 4, 10, 20, 35\}$ electrons. \rev{FermiNet was trained at each system size tested.}
  The dotted line marks the numerically exact energy for $N = 2$. Results are plotted using a symlog scale\rev{, whose linear region is shaded; because this scale expands the region near zero, error bars on points close to the reference are visually expanded relative to bars of equal magnitude further from it, and their sizes should be compared numerically rather than visually.}
  \rev{\textit{Lower:} the per-particle local energy variance $\mathrm{Var}[\Eloc]/N$ of each model, measured over the same inference batches. This is a reference-free measure of wavefunction quality, vanishing for an exact eigenstate. The error bars in the upper panel are defined as $\mathrm{stderr}/N = \sqrt{(\mathrm{Var}[\Eloc]/N)/(NM)}$.}
  \rev{Jastrow-Slater PDSG and FermiNet reach comparable energies at every system size, both below the CISD reference.}
  }
  \label{fig:energies}
\end{figure*}

Figure~\ref{fig:energies} reports converged per-particle energies relative to the CISD reference across system sizes, each obtained from a large-sample inference pass at the lowest-variance checkpoint (see Appendix~\ref{app:hyperparameters}). With a Slater base, accuracy generally improves from DS to PDS to FVF to PDSG. DS is by far the least accurate, with its gap above CISD widening as $N$ grows. The more expressive architectures give energies below the CISD reference at small $N$ and drift above it as $N$ increases, with PDSG remaining essentially at the reference for $N = 10$. \rev{At $N = 2$, the Slater-base PDS, FVF, and PDSG flows fall slightly below the numerically exact energy, by at most 2.7 standard errors. Because these models are evaluated on the same held-out base cache, their deviations are correlated and reflect a single statistical fluctuation. Note that their local energy variances are between 6 and 34 times that of the Jastrow-Slater PDSG model.}

Trained with a Jastrow-Slater base, the PDSG model improves upon the Slater-base PDSG model at every system size where both were tested, and converges below the CISD reference at all tested system sizes, up to $N = 35$. Because both the CISD and CNF energies are variational upper bounds on the exact ground-state energy (up to statistical error and ODE solver precision), an energy below CISD is a strictly better estimate. \rev{CISD is not size-consistent, and its accuracy as a reference may degrade at the larger system sizes studied; the FermiNet baseline provides an independent comparison at each $N$.} To provide a sense of scale, the Jastrow-Slater base energy lies just below the Hartree-Fock energy for all tested system sizes.

\rev{FermiNet also converges below the CISD reference at every system size, and the Jastrow-Slater PDSG model and FermiNet are on par across the range tested, both reaching the numerically exact energy at $N = 2$. Neither is consistently better: FermiNet reaches the lower variance at $N = 2$ and $N = 10$, and the Jastrow-Slater PDSG model at $N = 4$, 20, and 35. In each case the per-particle energy difference between the two is smaller than either model's improvement over CISD.}

\rev{The lower panel reports the per-particle local energy variance, which gives a measure of wavefunction quality without reference to CISD. Dividing by $N$ accounts for the growth of $\mathrm{Var}[\Eloc]$ with system size and makes the variance comparable across $N$. Even per particle, the variance increases with $N$ for every model, consistent with electron-electron correlation becoming more complex as the system grows.}

%%====================================================================
\section{Discussion}
\label{sec:discussion}
%%====================================================================

\paragraph{Physical interpretation.}
The four architectures represent distinct design choices for encoding interparticle correlations. DS aggregates all particle information through a single symmetric sum in $D$-dimensional space, giving the vector field no direct access to interparticle displacements; this limitation is most apparent at large $N$ with strong Coulomb repulsion, where DS converges to substantially higher energies than the other models. PDS addresses this by adding a pairwise stream with direct access to displacement vectors, enabling geometry-aware particle repulsion. FVF further adds expressivity, using multi-layer cross-stream communication with latent-space aggregation at each layer to iteratively refine particle representations. PDSG represents a structurally distinct approach, restricting the vector field to be the gradient of a scalar potential; it does, however, retain the pairwise stream of PDS and the latent-space aggregation of FVF. Experimentally, PDSG achieves the best accuracy among the CNF models, outperforming FVF despite the more constrained parameterization; FVF achieves the second-highest accuracy at the highest per-ODE-step cost. Notably, PDSG achieves the best accuracy with the fewest trainable parameters of the four architectures (Table~\ref{tab:param_count}). The strong performance of both relative to PDS is consistent with the benefit of latent-space aggregation over direct output-space pooling. 
The Jastrow-Slater base captures the short-range behavior near electron coalescence that the flow might otherwise struggle to represent, as evidenced by the significantly lower starting energy and improved accuracy at convergence compared to the Slater base alone. This motivates a general design principle for CNF-based VMC: the base wavefunction should incorporate as much known physical structure as possible (cusp conditions, nodal surfaces, long-range decay) so that the flow addresses only the residual correlations not captured analytically.

\paragraph{Connection to quantum-dot and cold-atom physics.}
The harmonically trapped spinless electron system studied here shares the Hamiltonian structure of quantum-dot and cold-atom systems, where confinement and Coulomb-type interactions arise naturally. A realistic quantum-dot simulation would additionally involve two-dimensional confinement, spin-carrying electrons, and an external magnetic field~\cite{reimann2002}; the flow framework could be extended to handle each of these.

The results demonstrate that CNF-based VMC matches or surpasses CISD reference energies for systems of up to $N = 35$ electrons, suggesting the approach is well-suited to the parameter sweeps (over $\omega$, $k$, or $N$) common in the study of strongly correlated quantum dots, where many independent ground-state calculations are required~\cite{reimann2002}. The offline sampling structure makes such sweeps especially efficient. The degree of correlation is set by the single dimensionless coupling $k/\sqrt{\omega}$, which in a quantum-dot context corresponds to the degree of Coulomb screening from the background material. A sweep over correlation strength can therefore be performed at fixed $\omega$ by varying the interaction strength $k$ alone; because the Slater base is the non-interacting ground state, which depends on $\omega$ but not $k$, a single pregenerated sample cache then serves the entire sweep, amortizing its generation cost across all values of $k$. The Jastrow-Slater base, by contrast, depends on $k$ as well as $\omega$, and must be regenerated for each value of $k$.

\paragraph{Limitations and scope.}
The most fundamental limitation of the flow ansatz concerns the nodal topology of the base wavefunction. As established in Section~\ref{sec:cusps}, the flow preserves nodal topology, so the base must match the nodal topology of an exact ground state. For a closed-shell particle number, the non-interacting ground state is a unique, nondegenerate Slater determinant, conjectured to share the nodal topology of the interacting ground state~\cite{nodal_topology}; this motivates the closed-shell values $N \in \{4, 10, 20, 35\}$ studied here. We also study $N = 2$, whose non-interacting ground state is degenerate but for which a single determinant still carries the correct nodal topology. For a general open-shell $N$, the exact ground state need not be expressible as a single Slater determinant, and a single-determinant base can have the wrong nodal topology; such systems would require a multi-determinant base to supply the correct topology, within which the flow could then refine the nodal surface. \rev{Such a base enters the framework without modification elsewhere: it need only admit offline sampling, be inexpensive enough to evaluate at every wavefunction evaluation, and supply the gradient and Hessian of $\log|\psi_{\mathrm{base}}|$ required by~\eqref{eq:laplacian_decomp}. Because the flow preserves nodal topology, the base determines not only the starting energy but also the accuracy attainable after refinement.}

The present work studies spinless (fully spin-polarized) electrons, which are representative of uniform-spin fermionic systems and avoid the complication of spin-up/spin-down channels. Extension to spin-$1/2$ electrons requires treating two species of fermions with separate antisymmetry; this can be done within the flow framework by using a block-diagonal Slater matrix in the base and modifying the permutation-equivariant vector field architecture to process spin-labeled inputs. For spin-$1/2$ electrons, the Kato cusp conditions differ for same-spin and opposite-spin pairs, so the Jastrow factor requires modification as well. Setting the interaction strength to $k=1$, the cusp parameters become $a = 1/4$ for parallel-spin coalescence and $a = 1/2$ for antiparallel-spin coalescence. \rev{The determinant vanishes at same-spin coincidences but not opposite-spin ones, so the Jastrow factor is more consequential in the opposite-spin channel. A spinful demonstration is left to future work.}
% The single-parameter Jastrow used here must therefore be replaced by a spin-dependent form with separate terms for each pair type.

The harmonic trap provides a clean test environment but lacks the electron-nuclear Coulomb singularities present in molecular systems. These singularities impose cusp conditions at each nuclear position, with cusp parameter proportional to the nuclear charge $Z$; for molecules with multiple nuclear species, $Z$ varies across sites. Satisfying these conditions in the base wavefunction is non-trivial: standard contracted Gaussian basis sets only approximate the correct short-range behavior near nuclei, and a dedicated electron-nuclear Jastrow factor, analogous to the electron-electron Jastrow used here, would be needed to enforce exact cusp conditions in the base. The nuclear cusp presents an additional complication compared to the electron-electron case: because the nuclear potential is attractive, electrons are frequently found near nuclei, and high-probability configurations lie precisely in the region where an unsatisfied cusp condition produces large local energy variance. In contrast, the repulsive electron-electron interaction suppresses configurations near electron-electron coalescence in the Born distribution, making electron-electron cusp violations less damaging in practice. Extending the flow training itself to reliably handle the molecular setting remains an open implementation challenge and an interesting direction for future work.

\paragraph{Connection to diffusion Monte Carlo.}
Fixed-node DMC achieves near-exact ground-state accuracy by applying imaginary-time projection within each nodal domain of a trial wavefunction; the residual fixed-node error depends entirely on the quality of the trial nodes, and with exact nodes the exact ground-state energy is recovered~\cite{foulkes2001}. A natural application of the CNF framework is to provide high-quality trial wavefunctions for DMC. The base wavefunction fixes the nodal topology, and the flow deforms the nodal surface within that topology to minimize the variational energy, producing a trial wavefunction with both correct nodal topology and well-optimized nodal domains. DMC subsequently refines the energy within the learned nodal domains. Applying DMC to CNF trial wavefunctions is therefore a natural direction for future work.

\paragraph{Future directions.}
Beyond the DMC application discussed above, further extensions of this work include: (i) richer correlated base wavefunctions, such as multi-determinant \rev{configuration interaction} expansions, which would provide a more accurate starting point for the flow and may be necessary for achieving the correct nodal topology for more general open-shell systems; (ii) more expressive equivariant vector field architectures, including attention-based vector field models \cite{bilos}; (iii) extension to molecules with electron-nuclear interactions \rev{and, in general, spin-$1/2$ electrons}, where the offline base-sample strategy is particularly attractive as the MCMC sampling cost for the base need only be paid once per molecular geometry; (iv) the periodic setting for solid-state materials, where the configuration space has toroidal rather than Euclidean topology and the base wavefunction would be built from Bloch orbitals\rev{; Bloch orbitals are normally complex~\cite{foulkes2001}, and the twist-averaged boundary conditions commonly used to reduce finite-size error~\cite{twist-averaged} admit no real solutions at a general twist, so the restriction to real wavefunctions adopted in Section~\ref{sec:vmc} would need to be relaxed}; (v) time-dependent vector fields, which are strictly more expressive than the time-independent architectures used here; whether this additional expressivity is necessary for more complex systems remains an open question.

%%====================================================================
\section{Conclusion}
\label{sec:conclusion}
%%====================================================================

We have presented a CNF-based framework for fermionic VMC that refines a fixed antisymmetric base wavefunction through a permutation-equivariant flow, with all base sampling performed offline. Three novel vector field architectures (Pairwise Deep Sets, FermiNet Vector Fields, and Pairwise Deep Sets Gradient) and an augmented dynamics formulation for efficient kinetic energy computation make the framework both accurate and computationally tractable.

Numerical experiments on harmonically trapped spinless electrons in three dimensions demonstrate ground-state energies surpassing CISD reference values for system sizes up to $N = 35$\rev{, on par with a FermiNet baseline}. Scaling experiments demonstrate near-ideal strong GPU scaling \rev{to} 128 A100s for \rev{the largest systems tested, up to} $N = 48$. Accuracy improved as the base incorporated more physical structure: the Jastrow-Slater base, which builds in the electron-electron cusp and short-range repulsion, reached lower energies than the bare Slater base. By refining an existing ansatz rather than replacing it, and by decoupling \rev{MCMC sampling from parameter optimization}, the framework offers a scalable route to improving established wavefunction methods with neural networks.
% We have presented a CNF-based framework for fermionic VMC in which offline base sampling enables embarrassingly parallel training without any MCMC at training time. Three novel permutation-equivariant vector field architectures (Pairwise Deep Sets, FermiNet Vector Fields, and Pairwise Deep Sets Gradient) and an augmented dynamics formulation for efficient kinetic energy computation make the framework both accurate and computationally tractable. The offline structure additionally supports expressive base wavefunctions, demonstrated here with a Jastrow-Slater base. In particular, choices of base with higher sampling cost relative to single Slater determinant wavefunctions are not precluded, since they only need to be sampled offline.

% Numerical experiments on harmonically trapped spinless electrons in three dimensions demonstrate ground-state energies surpassing CISD reference values for system sizes up to $N = 35$. Scaling experiments demonstrate near-ideal strong GPU scaling from 1 to 128 A100s for system sizes up to $N = 48$. The offline sampling structure decouples wavefunction expressivity from MCMC cost, an advantage that grows with system size as MCMC mixing times increase while ODE solves remain embarrassingly parallel.

\begin{acknowledgments}
%This manuscript was edited with AI assistance (Claude, Anthropic) for grammar and proofreading. All technical content and decisions are the authors' own.

%GPU training runs were carried out using compute credits provided through Modal's startup program. These credits supported all training and benchmarking experiments in this work other than the large-scale GPU scaling runs reported in Section~\ref{sec:scaling}.

%James Larsen acknowledges support from the Department of Energy Computational Science Graduate Fellowship under Award Number DE-SC0024386. This research used resources of the National Energy Research Scientific Computing Center, a DOE Office of Science User Facility supported by the Office of Science of the U.S. Department of Energy
%under Contract No. DE-AC02-05CH11231 using NERSC award NERSC DDR-ERCAP0038472.

We thank Modal Labs, Inc.\ for compute credits provided through its startup program, which supported all training and benchmarking experiments in this work other than the large-scale GPU scaling runs reported in \rev{Sections~\ref{sec:scaling} and~\ref{sec:weak_scaling}}; those runs used resources of the National Energy Research Scientific Computing Center (NERSC), a DOE Office of Science User Facility supported by the Office of Science of the U.S. Department of Energy under Contract No.~DE-AC02-05CH11231, using NERSC award DDR-ERCAP0038472. James Larsen acknowledges support from the Department of Energy Computational Science Graduate Fellowship under Award Number DE-SC0024386.
\end{acknowledgments}

\section*{Data Availability}
\rev{The data that supports the findings of this study, along with a static copy of the code used to generate them, are openly available at
\url{https://doi.org/10.5281/zenodo.23148849}. The code is also available at 
\url{https://github.com/sygaldry-tech/cnf-paper}.}

\bibliography{references}

%%====================================================================
%% Appendices
%%====================================================================
\appendix

\section{FermiNet Architecture and FermiNet Vector Fields}
\label{app:ferminet}

This appendix summarizes the FermiNet architecture~\cite{ferminet} and describes how it is adapted into the FermiNet Vector Field (FVF) introduced in Section~\ref{sec:ferminet_vf}. FermiNet processes a configuration $x = (x_1,\ldots,x_N)\in\R^{ND}$ through two coupled streams of learned features: a one-electron stream $h_i^{(l)}\in\R^{H_1}$ indexed by particle and a two-electron stream $h_{ij}^{(l)}\in\R^{H_2}$ indexed by ordered pair $(i,j)$.

\paragraph{Input features.}
Let $r_{ij} = x_j - x_i\in\R^D$. The streams are initialized by linear projections with \rev{an activation function $\sigma$}:
\begin{align}
  h_i^{(0)} &= \rev{\sigma}\bigl(W_1^{(0)}\,x_i\bigr), \label{eq:fn_init_one} \\
  h_{ij}^{(0)} &= \rev{\sigma}\bigl(W_2^{(0)}\,r_{ij}\bigr), \label{eq:fn_init_two}
\end{align}
where $W_1^{(0)}\colon\R^{D}\to\R^{H_1}$ and $W_2^{(0)}\colon\R^{D}\to\R^{H_2}$. \rev{In our implementation, the distances $\|x_i\|$ and $\|r_{ij}\|$ are appended to the inputs, with positions measured from the trap center.}

\paragraph{Interaction layers.}
Each layer $l=1,\ldots,N_L$ updates both streams. The one-electron update aggregates mean electron and pairwise information:
\begin{align}
  h_i^{(l)} = h_i^{(l-1)} +\rev{\sigma}\Bigl( &W_1^{(l)}\Bigl[h_i^{(l-1)},\;\frac{1}{N}\sum_{i'=1}^N h_{i'}^{(l-1)},\; \notag \\
  &\frac{1}{N}\sum_{j=1}^N h_{ij}^{(l-1)}\Bigr]\Bigr),
  \label{eq:fn_one_update}
\end{align}
where $[\cdot,\cdot,\cdot]$ denotes concatenation and $W_1^{(l)}\colon\R^{2H_1+H_2}\to\R^{H_1}$. The two-electron update is applied independently to each pair:
\begin{equation}
  h_{ij}^{(l)} = h_{ij}^{(l-1)} + \rev{\sigma}\bigl(W_2^{(l)}\,h_{ij}^{(l-1)}\bigr),
  \label{eq:fn_two_update}
\end{equation}
where $W_2^{(l)}\colon\R^{H_2}\to\R^{H_2}$. Permutation equivariance of $h_i^{(l)}$ follows by induction: permuting the input particles permutes the one- and two-electron features in the same way at every layer.

\paragraph{FermiNet wavefunction readout.}
After $N_L$ layers, the one-electron features are mapped to $N_{\mathrm{det}}$ sets of $N$ single-particle orbitals, and the wavefunction is formed as a sum of $N_{\mathrm{det}}$ determinants of these orbitals, with each orbital modulated by a decaying envelope~\cite{ferminet}. \rev{In our implementation, the envelope of orbital $j$ in determinant $k$ is the Gaussian $\exp(-\|A_{kj}\,x_i\|^2)$, with a trainable matrix $A_{kj}\in\R^{D\times D}$ initialized to reproduce the single-particle decay $\exp(-\omega\|x_i\|^2/2)$ of the harmonic trap.} Antisymmetry follows from the alternating property of the determinant: permuting particles permutes the rows of each determinant, changing its sign.

\paragraph{FermiNet Vector Fields.}
FVF retains the input embeddings~\eqref{eq:fn_init_one}--\eqref{eq:fn_init_two} and the $N_L$ interaction layers~\eqref{eq:fn_one_update}--\eqref{eq:fn_two_update}, with the particle configuration $x$ replaced by the trajectory state $y$ and \rev{the distance features omitted from the inputs. In place of} the orbital projection and determinant readout, \rev{FVF uses} a linear map directly to particle displacements:
\begin{equation}
  v_i(y) = W_{\mathrm{out}}\,h_i^{(N_L)},
  \label{eq:fvf_readout}
\end{equation}
where $W_{\mathrm{out}}\colon\R^{H_1}\to\R^D$. The output weights are initialized near zero so that the flow begins close to the identity. Permutation equivariance of $v$ is inherited from the equivariance of $h_i^{(N_L)}$.

\section{Laplacian Chain Rule Identity}
\label{app:laplacian_chain_rule}

\begin{lemma}[Laplacian chain rule]
\label{lem:laplacian_chain_rule}
Let $z\colon \R^n \to \R^m$ and $s\colon \R^m \to \R$ be smooth, and write $J = \partial z / \partial x \in \R^{m \times n}$ and $L_k = \Delta_x z_k$ for the coordinate Laplacians of $z$. Then
\begin{equation}
  \Delta_x[s(z(x))]
      = \tr(J^\top H_s J) + \nabla_z s \cdot L,
  \label{eq:laplacian_chain_rule}
\end{equation}
where $H_s = \partial^2 s / \partial z^2 \in \R^{m \times m}$ is the Hessian of $s$ at $z(x)$.
\end{lemma}

\begin{proof}
We apply the chain rule twice. The first-order chain rule gives
\begin{equation}
  \frac{\partial}{\partial x_j}s(z(x))
      = \sum_k \frac{\partial s}{\partial z_k}\frac{\partial z_k}{\partial x_j}.
\end{equation}
Differentiating again with respect to $x_j$ using the product rule on each term in the sum,
\begin{equation}
  \frac{\partial^2}{\partial x_j^2}s(z(x))
      = \sum_{k, l} \frac{\partial^2 s}{\partial z_k \partial z_l}
         \frac{\partial z_k}{\partial x_j}\frac{\partial z_l}{\partial x_j}
         + \sum_k \frac{\partial s}{\partial z_k}
           \frac{\partial^2 z_k}{\partial x_j^2}.
\end{equation}
Summing over $j = 1,\ldots,n$ gives $\Delta_x[s(z(x))]$ on the left. Substituting $J_{kj} = \partial z_k / \partial x_j$,
\begin{align}
  \Delta_x[s(z(x))]
      =& \sum_{k,l}(H_s)_{kl} \sum_j J_{kj} J_{lj} \notag \\
      &+ \sum_k (\nabla_z s)_k \sum_j \frac{\partial^2 z_k}{\partial x_j^2}.
\end{align}
The inner sum $\sum_j J_{kj}J_{lj}$ is the $(k,l)$ entry of $JJ^\top \in \R^{m\times m}$, and $\sum_j \partial^2 z_k/\partial x_j^2 = \Delta_x z_k = L_k$ by definition, giving
\begin{equation}
  \Delta_x[s(z(x))]
      = \sum_{k,l}(H_s)_{kl}(JJ^\top)_{kl} + \nabla_z s \cdot L.
\end{equation}
The double sum $\sum_{k,l}(H_s)_{kl}(JJ^\top)_{kl}$ is the Frobenius inner product of $H_s$ with $JJ^\top$, equal to $\tr(H_s^\top JJ^\top)$. Since $s$ is smooth, its mixed partial derivatives commute ($\partial^2 s/\partial z_k \partial z_l = \partial^2 s/\partial z_l \partial z_k$ for all $k,l$), so $H_s$ is symmetric and $H_s^\top = H_s$, reducing this to $\tr(H_s JJ^\top)$. The cyclic property of trace then gives $\tr(H_s JJ^\top) = \tr(J^\top H_s J)$, yielding~\eqref{eq:laplacian_chain_rule}.
\end{proof}

\section{Proof of Proposition~\ref{prop:augmented_dynamics}}
\label{app:augmented_dynamics}

\begin{proof}
The proof proceeds by differentiating each augmented state variable with respect to $t$, exchanging the order of time and spatial differentiation, and applying the chain rule identities of Appendix~\ref{app:laplacian_chain_rule}. We treat each equation of the proposition in turn.

\medskip\noindent\textit{Equation~\eqref{eq:aug_dy}}:
The forward ODE~\eqref{eq:forward_ode} maps a base sample $z$ to $x = f(z)$ by integrating $\dot{y} = v(y)$ from $t = 0$ to $t = 1$ with $y(0) = z$. Integrating the same ODE backward from $t = 1$ to $t = 0$ with initial condition $y_1(1) = x$ reverses this trajectory, giving $y_1(0) = z = f^{-1}(x)$.

\medskip\noindent\textit{Equation~\eqref{eq:aug_dJ}}:
$y_2 = \partial y_1/\partial x$ is the Jacobian of $y_1$ with respect to $x$. At $t = 1$, $y_2(1) = \partial x/\partial x = I$. Differentiating both sides of $\dot{y}_1 = v(y_1)$ with respect to $x$ and exchanging the order of $\mathrm{d}/\mathrm{d}t$ and $\partial/\partial x$ yields
\begin{equation}
  \frac{\mathrm{d}y_2}{\mathrm{d}t}
      = \frac{\partial}{\partial x}\frac{\mathrm{d}y_1}{\mathrm{d}t}
      = \frac{\partial}{\partial x}v(y_1)
      = \frac{\partial v}{\partial y_1}\frac{\partial y_1}{\partial x}
      = J_v\,y_2.
\end{equation}

\medskip\noindent\textit{Equation~\eqref{eq:aug_dL}}:
$y_3$ collects the coordinate Laplacians of $y_1$ with respect to $x$: $(y_3)_k = \Delta_x(y_1)_k$. At $t = 1$, $y_1(1) = x$, so $(y_3(1))_k = \Delta_x x_k = 0$. Differentiating with respect to $t$ and exchanging $\mathrm{d}/\mathrm{d}t$ with $\Delta_x$ gives
\begin{equation}
  \frac{\mathrm{d}(y_3)_k}{\mathrm{d}t}
      = \Delta_x\frac{\mathrm{d}(y_1)_k}{\mathrm{d}t}
      = \Delta_x[v_k(y_1)].
\end{equation}
Applying Lemma~\ref{lem:laplacian_chain_rule} yields
\begin{equation}
  \Delta_x[v_k(y_1)]
      = \tr(y_2^\top H_{v_k} y_2) + \nabla_{y_1} v_k \cdot y_3.
\end{equation}
Since $\nabla_{y_1} v_k = (J_v)_{k,:}$ is the $k$-th row of $J_v$, the dot product $\nabla_{y_1} v_k \cdot y_3 = (J_v)_{k,:}\,y_3$. Stacking over $k = 1,\ldots,ND$ yields~\eqref{eq:aug_dL}.

\medskip\noindent\textit{Equation~\eqref{eq:aug_dlambda}}:
$y_4$ tracks $\log|\det(\partial y_1/\partial x)|$, the log-determinant of the Jacobian of $y_1$ with respect to $x$. At $t = 1$, $\partial y_1(1)/\partial x = \partial x/\partial x = I$, so $y_4(1) = \log|\det I| = 0$. Applying~\eqref{eq:inst_logdet} with $y = y_1$ and initial condition $z = x$ gives $\dot{y}_4 = \tr J_v$.

\medskip\noindent\textit{Equation~\eqref{eq:aug_dgamma}}:
$y_5 = \nabla_x y_4$ is the gradient of the log-determinant with respect to $x$. At $t = 1$, $y_4(1) = 0$, so $y_5(1) = \nabla_x 0 = 0$. Exchanging $\mathrm{d}/\mathrm{d}t$ and $\nabla_x$ and substituting $\dot{y}_4 = \tr J_v$,
\begin{equation}
  \frac{\mathrm{d}y_5}{\mathrm{d}t}
      = \nabla_x\frac{\mathrm{d}y_4}{\mathrm{d}t}
      = \nabla_x(\tr J_v(y_1(x))).
\end{equation}
Applying the gradient chain rule gives
\begin{align}
  \nabla_x(\tr J_v(y_1))
      &= \left(\frac{\partial y_1}{\partial x}\right)^\top \nabla_{y_1}(\tr J_v) \notag \\
      &= y_2^\top \nabla_{y_1}(\tr J_v).
\end{align}

\medskip\noindent\textit{Equation~\eqref{eq:aug_dl}}:
$y_6 = \Delta_x y_4$ is the Laplacian of the log-determinant with respect to $x$. At $t = 1$, $y_4(1) = 0$, so $y_6(1) = \Delta_x 0 = 0$. Exchanging $\mathrm{d}/\mathrm{d}t$ and $\Delta_x$ and substituting $\dot{y}_4 = \tr J_v$, we get
\begin{equation}
  \frac{\mathrm{d}y_6}{\mathrm{d}t}
      = \Delta_x\frac{\mathrm{d}y_4}{\mathrm{d}t}
      = \Delta_x(\tr J_v(y_1(x))).
\end{equation}
Applying Lemma~\ref{lem:laplacian_chain_rule} gives
\begin{equation}
  \Delta_x(\tr J_v(y_1))
      = \tr(y_2^\top H_{\tr J_v} y_2) + \nabla_{y_1}(\tr J_v)\cdot y_3.
\end{equation}
The values $y_1(0),\ldots,y_6(0)$ produced by integrating this system from $t=1$ to $t=0$ are the six quantities of~\eqref{eq:laplacian_decomp_pieces}; together with $\nabla_z s$ and $H_s$ evaluated at $z = y_1(0)$, they suffice to evaluate the local kinetic energy via~\eqref{eq:laplacian_decomp} and~\eqref{eq:kinetic_local}.
\end{proof}

\section{CISD Reference Energies}
\label{app:cisd}

\begin{figure*}
    \centering
    \includegraphics[width=\linewidth]{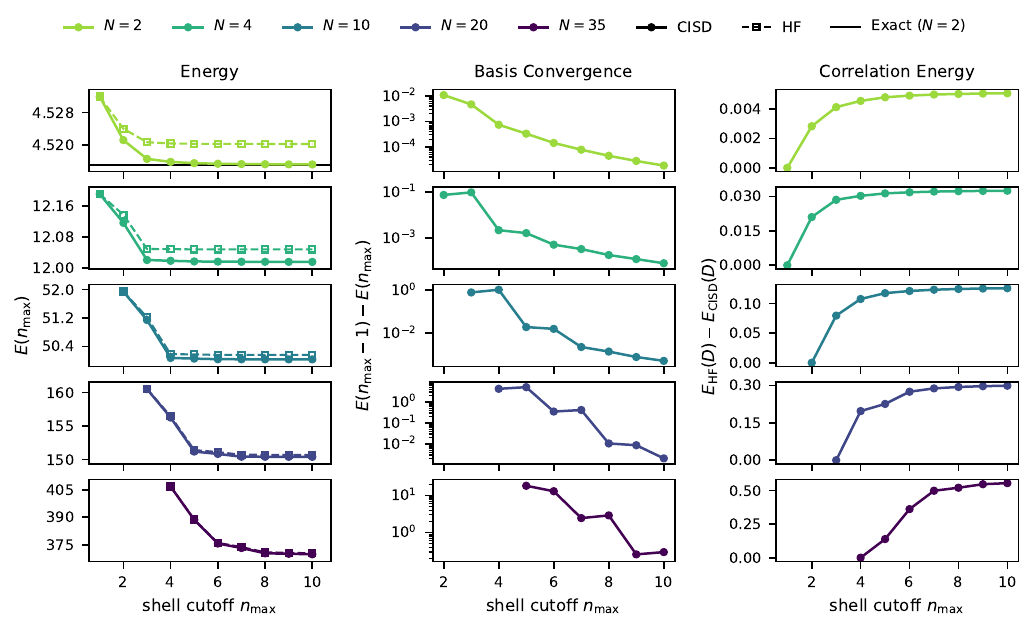}
    \caption{Convergence of the CISD reference energies with shell cutoff $n_{\max}$, for each system size $N$. \textit{Left:} energy $E(n_{\max})$ for CISD (solid) and Hartree-Fock (dashed), with the numerically exact energy for $N = 2$ shown as a solid black line. \textit{Center:} the successive decrease $E(n_{\max}-1) - E(n_{\max})$ of the CISD energy, showing its convergence with basis size. \textit{Right:} the correlation energy $E_\mathrm{HF}(n_{\max}) - E_\mathrm{CISD}(n_{\max})$ recovered by CISD relative to the correlation-free Hartree-Fock energy, at each shell cutoff $n_{\max}$.}
    \label{fig:cisd_convergence}
\end{figure*}

The reference energies used throughout Section~\ref{sec:experiments} are obtained using \rev{CISD}, built on integral and solver routines from PySCF~\cite{pyscf}. This appendix describes the finite basis in which those calculations are performed and how it relates to the harmonic oscillator spectrum.

\paragraph{Basis set.}
% The CISD calculations use the trap frequency $\omega = 1$ and interaction strength $k = 1$ of our experiments.
We define the single-particle basis so that it spans exactly the space of the lowest harmonic oscillator solutions. Each oscillator eigenstate~\eqref{eq:ho_orbitals} is a product of Hermite polynomials in the coordinates $x_1, x_2, x_3$ and the fixed Gaussian $e^{-\omega |x|^2 / 2}$; since the Hermite polynomial $H_{n_d}$ has degree $n_d$, the state $\phi_n$ is a polynomial of total degree $|n| = n_1 + n_2 + n_3$ times that Gaussian. Ordering the oscillator states by $|n|$ into shells, the span of all states through shell $n_{\max}$ is the set of functions $p(x)\,e^{-\omega|x|^2/2}$ with $p$ a polynomial of total degree at most $n_{\max}$; this holds because the leading term of $\prod_d H_{n_d}$ is the monomial $\prod_d x_d^{n_d}$, so the products with $|n| \le n_{\max}$ and the monomials of degree $\le n_{\max}$ span the same polynomial space.

We reproduce this space with a Cartesian Gaussian basis, whose primitives are Cartesian monomials times a Gaussian, $x_1^{n_1} x_2^{n_2} x_3^{n_3}\, e^{-\alpha |x|^2}$. Fixing the exponent $\alpha = \omega/2$ to match~\eqref{eq:ho_orbitals} and including every monomial of total degree $|n| \le n_{\max}$, the resulting basis spans precisely the polynomials of degree $\le n_{\max}$ times the Gaussian, and hence the same space as oscillator shells $0$ through $n_{\max}$. We use Cartesian Gaussians because they are the standard primitives of quantum-chemistry integral codes, which lets us evaluate the required one- and two-electron integrals directly with PySCF~\cite{pyscf}. The number of Cartesian monomials of total degree exactly $|n|$ in three variables is $(|n|+1)(|n|+2)/2$, so the basis through degree $n_{\max}$ contains
\begin{align}
    N_{\mathrm{orb}} &= \sum_{|n|=0}^{n_{\max}} \frac{(|n|+1)(|n|+2)}{2} \notag \\
    &= \frac{(n_{\max}+1)(n_{\max}+2)(n_{\max}+3)}{6}
\end{align}
orbitals, equal to the number of oscillator states through shell $n_{\max}$. Table~\ref{tab:cisd_basis} lists the orbital count by cutoff through $n_{\max}=10$.

\begin{table}[h]
  \centering
  \caption{Number of basis orbitals by shell cutoff $n_{\max}$, equal to the number of harmonic oscillator states through shell $n_{\max}$.}
  \label{tab:cisd_basis}
  \small
  \setlength{\tabcolsep}{5pt}
  \begin{tabular}{lcccccccccc}
    \toprule
    $n_{\max}$ & 1 & 2 & 3 & 4 & 5 & 6 & 7 & 8 & 9 & 10 \\
    \midrule
    $N_{\mathrm{orb}}$ & 4 & 10 & 20 & 35 & 56 & 84 & 120 & 165 & 220 & 286 \\
    \bottomrule
  \end{tabular}
\end{table}

\paragraph{CISD energies.}
The trapped fermions are spinless, so we run CISD in the fully spin-polarized sector, where all electrons share one spin and the many-body state is antisymmetric under exchange, matching the physics of identical spinless fermions. At fixed $n_{\max}$, the CISD energy is a variational upper bound on the exact ground-state energy in that basis. For $N = 2$, singles and doubles already exhaust the full configuration space, so CISD coincides with \rev{full configuration interaction}, which is exact within the finite basis.
% The dominant source of error is therefore basis-set incompleteness rather than the truncation of the excitation level. 
The reference energies reported in the main text use $n_{\max} = 10$ ($N_{\mathrm{orb}} = 286$). Figure~\ref{fig:cisd_convergence} shows CISD energies converging with $n_{\max}$ for each system size.

\section{Experimental Hyperparameters}
\label{app:hyperparameters}

This appendix collects the full architectural and training hyperparameters for all models evaluated in Section~\ref{sec:experiments}.

\paragraph{CNF architectures.}
DS used hidden dimension 32 for $g$ and $h$, each with 4 hidden layers. PDS used hidden dimension 32 for $g$ and $h$ and hidden dimension 16 for $p$, each with 4 hidden layers. FVF used one-electron stream hidden dimension 32, two-electron stream hidden dimension 16, and 4 interaction layers. PDSG used hidden dimension 32 for the single-particle embedding $q$ and 16 for the pairwise embedding $p$, each embedding with 3 hidden layers; each scalar readout used a single hidden layer, of width 32 for $a_q$ and 16 for $a_p$. The ODE solves for both CNF directions were carried out using the Tsit5 explicit Runge-Kutta method as implemented in Diffrax~\cite{diffrax}, with step sizes chosen adaptively by a PID controller so that the estimated local error at each step stays within the relative and absolute tolerances $\mathrm{rtol} = 10^{-7}$ and $\mathrm{atol} = 10^{-9}$. Total trainable parameter counts for each model are listed in Table~\ref{tab:param_count}.

\begin{table}[h]
  \centering
  \caption{Total trainable parameter count for each model. \rev{The FermiNet count grows linearly with $N$ through its orbitals and envelopes.}}
  \label{tab:param_count}
  \small
  \setlength{\tabcolsep}{5pt}
  \begin{tabular}{lccccc}
    \toprule
    & DS & PDS & FVF & PDSG & \rev{FermiNet} \\
    \midrule
    Parameters & 6{,}790 & 7{,}721 & 11{,}747 & 4{,}226 & \rev{\begin{tabular}{@{}c@{}}$424{,}160$ \\ $+\,4{,}240N$\end{tabular}} \\
    \bottomrule
  \end{tabular}
\end{table}

\begin{revtext}
\paragraph{FermiNet baseline.}
The FermiNet baseline used $N_L = 4$ layers, stream widths $H_1 = 256$ and $H_2 = 32$, and $N_\mathrm{det} = 16$ determinants~\cite{ferminet}. Before variational training, the orbitals were pretrained to fit the noninteracting Slater determinant. Sampling followed the FermiNet repository~\cite{ferminet_github}, with one persistent chain per batch sample and the step size adapted toward a 50\% acceptance rate; between parameter updates, each chain took 30 all-electron moves, following Ref.~\onlinecite{vonglehn2023psiformer}. Training used the same MinSR settings, learning rate schedule, and convergence criterion as the CNF models. Final energies were estimated with the same MCMC routine used to generate the base sample caches, so that error bars for both models are produced by the same procedure.
\end{revtext}

\paragraph{Base sample cache generation.}
Samples from the base distributions were drawn using 4 GPUs per system, with 16 parallel Metropolis-Hastings chains per GPU. The step size was adapted during burn-in to target a 23.4\% acceptance rate; burn-in length was determined via the Gelman-Rubin $\hat{R}$ statistic (threshold 1.01), and the thinning factor was set adaptively from the integrated autocorrelation time $\tau$.

\paragraph{Training.}
The four vector field architectures were compared on the Slater base at $N \in \{2, 4, 10\}$. The Jastrow-Slater base was evaluated with the PDSG vector field, the best-performing architecture, which was trained for $N \in \{2, 4, 10, 20, 35\}$. The number of GPUs was scaled with system size: 1 GPU for $N = 2$ and $N = 4$, 2 GPUs for $N = 10$, 4 GPUs for $N = 20$, and 8 GPUs for $N = 35$.

MinSR was applied with relative pseudo-inverse cutoff $10^{-6}$, keeping the effective condition number below $10^6$. Parameter update norms were clipped to be at most 1.0.
The learning rate followed an inverse-time decay schedule $\eta(t) = \eta_0 / (1 + t / t_0)$ with $\eta_0 = 0.05$ and $t_0 = 25$.
We check convergence by fitting lines to the local energy mean and the log local-energy variance over a trailing window of $w = 100$ iterations. For each quantity, the fitted slope is multiplied by $w$ to give the total drift over the window; for the energy this is further divided by the mean per-particle energy magnitude over the window, yielding a dimensionless scaled slope $m_E$, while the log-variance slope is already a fractional rate of change and requires no normalization, yielding $m_V$ directly. Training is converged when both $m_E, m_V > -0.01$ for five consecutive iterations, indicating that neither the energy nor the variance is still meaningfully decreasing. 
Once convergence is declared, an inference pass selects the checkpoint within the final $100$ iterations with the lowest local energy variance. The energy of this checkpoint is then estimated from an inference batch of $409{,}600$ samples drawn from $\Pi_\theta$\rev{, obtained by pushing base samples from a separate held-out cache through the flow. The held-out cache is generated by an independent MCMC run and never used in training, so no inference sample was seen during optimization. Reported standard errors were computed assuming independent samples; although the base samples were thinned at the estimated autocorrelation time, residual correlations may remain, so these standard errors are lower bounds.} These estimates are the converged energies reported in Figure~\ref{fig:energies}. All computations use double-precision floating-point arithmetic throughout. The implementation is built using JAX~\cite{jax}, Equinox~\cite{equinox}, Diffrax~\cite{diffrax}, and folx~\cite{gao2023folx}.

\end{document}